\documentclass[aps,prl,twocolumn,superscriptaddress,floatfix,nofootinbib,showpacs,longbibliography,groupedaddress]{revtex4-2}
\usepackage{comment}
\usepackage{amsmath,amssymb,amsthm,bm,amsfonts,mathrsfs,bbm}
\usepackage[colorlinks=true, citecolor=blue, urlcolor=blue]{hyperref}
\usepackage[caption = false]{subfig}
\usepackage{times}
\usepackage{xcolor,colortbl}
\usepackage{ragged2e}
\usepackage{tikz}
\usepackage{scalerel}
\usetikzlibrary{calc}
\usetikzlibrary {arrows.meta}
\usepackage{leftindex}
\usepackage{braket}

\definecolor{abc}{rgb}{.5,0.,.5}
\definecolor{green1}{rgb}{0,.8,.3}

\newcommand{\tr}{\operatorname{Tr}}

\newtheorem{theorem}{Theorem}

\newtheorem{definition}{Definition}
\newtheorem{proposition}{Proposition}

\newtheorem{lemma}{Lemma}

\begin{document}

\title{Birkhoff Violation in Even-Polygon General Probabilistic Theories: Implications for State Transformations and Information Processing}

\author{Ananya Chakraborty}
\email{ananyaphys.c@gmail.com}
\affiliation{Department of Physics of Complex Systems, S. N. Bose National Center for Basic Sciences, Block JD, Sector III, Salt Lake, Kolkata 700106, India.}

\author{Sahil Gopalkrishna Naik}
\affiliation{Department of Physics of Complex Systems, S. N. Bose National Center for Basic Sciences, Block JD, Sector III, Salt Lake, Kolkata 700106, India.}

\author{Samrat Sen}
\affiliation{Department of Physics of Complex Systems, S. N. Bose National Center for Basic Sciences, Block JD, Sector III, Salt Lake, Kolkata 700106, India.}

\author{Ram Krishna Patra}
\affiliation{Department of Physics of Complex Systems, S. N. Bose National Center for Basic Sciences, Block JD, Sector III, Salt Lake, Kolkata 700106, India.}

\author{Pratik Ghosal}
\affiliation{Department of Physics of Complex Systems, S. N. Bose National Center for Basic Sciences, Block JD, Sector III, Salt Lake, Kolkata 700106, India.}

\author{Mir Alimuddin}
\affiliation{Department of Physics of Complex Systems, S. N. Bose National Center for Basic Sciences, Block JD, Sector III, Salt Lake, Kolkata 700106, India.}

\author{Manik Banik}
\affiliation{Department of Physics of Complex Systems, S. N. Bose National Center for Basic Sciences, Block JD, Sector III, Salt Lake, Kolkata 700106, India.}

\begin{abstract}
The entropy principle asserts that spontaneous processes drive systems toward states of higher disorder. Operational origin of such randomness is a question of utmost importance. The seminal Birkhoff theorem (also known as Birkhoff–von Neumann theorem) identifies doubly stochastic transformations as convex mixtures of reversible transformations, thereby providing an epistemic account of disorder. Quantum theory, however, does not satisfy this equivalence in general. In this work, we investigate this question within the framework of Generalized Probabilistic Theories (GPTs). We prove a violation of the Birkhoff theorem in the square-bit model, a canonical non-quantum GPT, and further show that this violation persists asymptotically under tensor products of any finite, yet arbitrarily large, number of copies of even-polygon models. This provides evidence that the failure of the Birkhoff theorem is an intrinsic operational feature of a wide class of GPTs, rather than a peculiarity of quantum theory. Furthermore, we show that, beyond quantum theory, the preorder over states induced by mixtures of reversible transformations can differ from that induced by doubly stochastic transformations, leading to inequivalent notions of state convertibility and distinct resource theories of purity. Beyond its foundational significance, we show that the violation of the Birkhoff theorem admits a concrete operational interpretation: in the square-bit theory, it serves as a genuine resource for enhancing the performance of Random Access Codes. Finally, we discuss the construction of entropy functions in GPTs that remain consistent with thermodynamic principles, in the spirit of von Neumann’s formulation.
\end{abstract}

\maketitle
{\it Introduction.--} The mathematical formulation of physical theories assigns states to systems and specifies admissible dynamical transformations consistent with underlying principles. A natural setting is that of an isolated system subject to conserved quantities. Such systems are described by the microcanonical state, defined as the uniform distribution over all configurations compatible with the constraints. This state represents maximal ignorance beyond the conserved quantities and serves as a canonical free state in operational approaches to thermodynamics. Bistochastic transformations plays a central role in this context due to its connection with the second law, which asserts that the entropy of an isolated system does not decrease. At the microscopic level, however, the dynamics is reversible, and entropy is invariant under reversible transformations. The emergence of entropy increase is therefore attributed to coarse-graining: when the implemented reversible transformation is not specified, the effective dynamics is described by a convex mixture of reversible operations, yielding a bistochastic transformation. This motivates the interpretation that entropy increase arises from incomplete information about the underlying dynamics. A fundamental question is whether this randomness is purely epistemic or reflects intrinsic features of the theory.

In classical theory, the Birkhoff–von Neumann theorem \cite{Birkhoff1946, Neumann1953} establishes that every bistochastic transformation \cite{Marshall2011} can be expressed as a convex combination of reversible transformations. This provides an operational realization of bistochastic dynamics in terms of random reversible transformation and supports an epistemic interpretation of disorder. Various extensions of this result have been obtained for infinite-dimensional settings \cite{Isbell1955, Losert1982, Punescu2017}. In contrast, quantum theory does not satisfy this equivalence in general: except in dimension two, unital channels cannot be represented as convex mixtures of unitary transformations \cite{Landau1993, Mendl2009, Ivan2011}. Although an asymptotic version of the equivalence was conjectured \cite{Smolin2005}, it is now known that violations persist even asymptotically \cite{Haagerup2011, Haagerup2015}. 
This raises the question of whether Birkhoff violation is a uniquely quantum phenomenon or a generic feature of broader operational theories.
In this direction the authors in \cite{Chiribella2017} have extended the framework of microcanonical thermodynamics from classical and quantum theories to GPTs. They introduced a resource-theoretic framework based on three natural classes of free operations: Random Reversible (RaRe) operations, given by convex mixtures of reversible transformations; Noisy operations, obtained by coupling to ancillas prepared in the microcanonical state, applying reversible interactions, and discarding subsystems; and Unital operations, which preserve the microcanonical state. For a broad class of GPTs, they established the inclusion
\begin{align*}
\mathrm{RaRe}\subseteq \mathrm{Noisy}\subseteq \mathrm{Unital},    
\end{align*}
and showed that state convertibility under unital operations is completely characterized by majorization. By contrast, majorization is, in general, only a necessary condition for state convertibility under RaRe and noisy operations.\\
Motivated by the violation of the Birkhoff--von Neumann theorem in the square-bit theory, we investigate the relationship between RaRe and unital transformations in this non-classical GPT. The square-bit, which constitutes the local state space of the Popescu--Rohrlich (PR) theory \cite{Popescu1994,Popescu2014}, exhibits a striking departure from quantum theory. While in quantum theory RaRe and unital operations induce the same state-convertibility relation despite the failure of the Birkhoff--von Neumann theorem, we show that this equivalence breaks down in the square-bit model. We choose the square-bit model because it is one of the simplest non-quantum Generalized Probabilistic Theories, making it an ideal framework for exploring structural features beyond quantum theory and demonstrating that the observed phenomenon of Birkhoff violation is a broader structural property of GPTs rather than a peculiarity of quantum theory alone. Thus, beyond the general structural results of Ref.~\cite{Chiribella2017}, our work provides an explicit operational manifestation of Birkhoff violation in a concrete GPT. The main contributions of this work are as follows:
\begin{itemize}
\item We completely characterize reversible, Random Reversible (RaRe), and unital transformations in the square-bit theory, providing explicit criteria for classifying arbitrary transformations.
\item We derive necessary and sufficient conditions for state convertibility under these classes of transformations and identify the corresponding complete monotones.
\item We demonstrate that the Birkhoff--von Neumann violation has operational consequences by analyzing Random Access Code tasks \cite{Wiesner1983,Ambainis2002}, thereby quantifying the power of different classes of transformations for information processing.
\item We prove that the separation between RaRe and unital operations persists asymptotically, showing that the Birkhoff-von Neumann violation is not merely a single-copy phenomenon.
\item We introduce an operationally consistent notion of entropy for the square-bit theory.
\end{itemize}
\color{black}
From a thermodynamic perspective, the preparation of pure states entails nontrivial costs \cite{Landauer1961, Bennett1982, Faist2018, Guryanova2020}, and pure states can be regarded as a resource enabling work extraction, as in the models of Leo Szilard and Charles H. Bennett \cite{Szilard1929, Bennett1993, Maruyama2009}. In quantum theory, this motivates the resource theory of purity, where the maximally mixed state is free and convex mixtures of unitary operations define the free transformations \cite{Horodecki2003(1), Horodecki2003(2)}. Although, in dimensions greater than two, the set of unital maps strictly contains mixtures of unitaries, the corresponding state transformation criteria coincide \cite{Gour2015, Chiribella2017}. We show that this equivalence fails in the square-bit model. We have systematically characterized all valid transformations and derived the corresponding state-transformation criteria in terms of entropic monotones. As a consequence, violations of Birkhoff theorem give rise to in-equivalent resource theories of purity, highlighting a fundamental separation between different classes of free operation. The operational implications of the Birkhoff–von Neumann theorem manifest in distinct ways, particularly in information-processing tasks such as Random Access Codes, where different classes of allowed transformations lead to fundamentally different performance limits.
 Finally, we introduce an entropy function within the GPT framework that is consistent with thermodynamic requirements and retains the conceptual structure of the construction due to John von Neumann.\\\\
{\it Generalised Probability Theory.--}\label{GPT}
The Generalized Probabilistic Theories (GPT) framework encompasses all operational theories that describe physical systems using states to determine the probabilities of measurement outcomes. Originating in the 1960s, its foundational goal was to provide an axiomatic derivation of quantum mechanics \cite{Mackey1963,Ludwig1967,Davies1970,Mielnik1974}. More recently, fueled by developments in quantum information theory, the GPT framework has gained renewed attention \cite{Hardy2001,Barrett2007,Chiribella2011,Barnum2011,Masanes2011,Naik2022,Patra2023,Patra2026}. A system in a GPT is specified by a triple $(\mathbf{S}, \mathbf{E}, \mathbf{T})$, representing the sets of normalized states, effects, and allowed transformations. Physical principles impose constraints on these sets, yielding a powerful and elegant language for analyzing and comparing different physical theories.
In a GPT, the state-space $\mathbf{S}$ is assumed to be a convex-compact set embedded in some real vector space $\mathbb{V}$. Given two states $\omega_1,\omega_2\in{\bf S}$, convexity ensures preparing their statistical mixture,
\begin{align*}
   p\omega_1+(1-p)\omega_2 \in {\bf S}.
\end{align*} States that cannot be expressed in terms of convex mixture of other states are called the extreme states. On the other hand, no physical distinction is there between states that can be prepared exactly and states that can be prepared to arbitrary accuracy, which is captured through the assumption that the set $\mathbf{S}$ is topologically closed. To avoid mathematical intricacy here we will assume $\mathbb{V}$ to be finite dimensional, albeit the framework is general enough to include the infinite dimensional cases too.
Effects ${\bf E}$ are linear functionals that map each normalized state to probabilities, i.e., 
\begin{align*}
e(\omega)\in[0,1],~\forall~\omega\in\mathbf{S}~\&~e\in\mathbf{E}.
\end{align*}Normalization of the states are specified by the unit effect $u$, i.e., $u(\omega)=1,~\forall\omega\in\mathbf{S}$. It is often convenient to consider unnormalized states and effects that form positive cones. The effect cone ${\bf V}_+^\star$ is dual to the state cone.
\begin{align*}
 {\bf V}_+:=\{\lambda\omega~|~\lambda\in\mathbb{R}_{\ge}~\&~\omega\in{\bf S}\} . 
\end{align*}
Collection of effects adding up to the unit effect forms a measurement, i.e., $M\equiv\{e_i~|~\sum e_i=u\}$. A set of states $\{\omega_j\}\subset{\bf S}$ will be perfectly distinguishable if there exist a measurement $M\equiv\{e_i\}$ such that $e_i(\omega_j)=\delta_{ij}$.\\
A dynamical evolution is defined as a channel $T\in{\bf T}$ that maps normalized states to normalized states. It is also natural to demand that its action on a part of a larger system to result in valid composite states. For more detailed of the framework we refer the recent review \cite{Plvala2023} (see also references therein). We end this section by pointing out that the description of a quantum system associated with Hilbert space $\mathcal{H}$ perfectly fits within this framework, with $(\mathbf{S},\mathbf{E},\mathbf{T})\equiv(\mathbf{D}(\mathcal{H}),\mathbf{P}(\mathcal{H}),\mathbf{Ch}(\mathcal{H}))$. Here, $\mathbf{D}(\mathcal{H})$ and $\mathbf{P}(\mathcal{H})$ respectively denotes the set of density operators and the set of positive operators (bounded above by the identity operators) acting on $\mathcal{H}$, and $\mathbf{Ch}(\mathcal{H})$ denotes the set of all completely positive trace-preserving (CPTP) maps from space of linear operators $\mathcal{L}(\mathcal{H})$ to $\mathcal{L}(\mathcal{H})$. While the description of classical and quantum theories fits perfectly within this framework. However, it encompasses a broad spectrum of other models  \cite{Plvala2023}. Here, we briefly revisit a toy model having a very simple mathematical description, yet offering rich operational insights \cite{Janotta2011,Muller2012,Janotta2013,Massar2014,Banik2019,Bhattacharya2020,Saha2020}.\\\\
{\it Square-bit model.--}
From operational point of view,  states in a GPT are completely specified by outcome statistics of a finite set of measurements called the `fiducial' set \cite{Hardy2001}. For instance, the state space of a $d$-level classical system is a $d$-simplex, with the fiducial set consisting of a single $d$-outcome measurement. In contrast, the three Pauli measurements constitute a fiducial set for the two-level quantum system (qubit). On the other hand, for square-bit, fiducial set consists of two dichotomic measurements $M_{13}\equiv\{e_1,e_3\}$ and $M_{24}\equiv\{e_2,e_4\}$, which are incompatible in the sense that their outcome cannot be simultaneously obtained from a single joint measurement. A generic state $\omega\in\mathbf{S}_{\bm\square}$ is thus specified as $\omega\equiv(p,1-p|q,1-q)^\intercal$, where $p:=e_1(\omega)$ and $q:=e_2(\omega)$. While probabilities of Pauli measurements on qubit are fundamentally constrained through uncertainty relation \cite{Maassen1988}, in case of square-bit no such restriction is imposed, and thus implying the range $p,q\in[0,1]$. The state $\omega$ can be represented as a vector $\omega\equiv(2p-1,2q-1,1)^\intercal\in\mathbb{R}^3$. Accordingly, the state-space $\mathbf{S}_{\bm\square}$ turns out to be the convex hull of four extreme points, and the effect space $\mathbf{E}_{\bm\square}$ turns out to be the convex hull of six extreme effects (see Fig.\ref{fig1}). 
\begin{figure}[h!]
\centering
\begin{tikzpicture}[scale=1]
\coordinate (A) at (0,1,0);
\coordinate (H) at (0,5,0);
\coordinate (H1) at (0,2.1,0);
\draw[dashed] (A)--(0,5,0);s
\draw[->,thick] (0,5,0)--(0,6.5,0);
\draw[->,thick] (A)--(2,1,0);
\draw[->,thick] (A)--(0.5,1.4,3.2);
\node at (0.3,6.1,0) {\scalebox{1.3}{$Z$}};
\node at (1.5,0.25,-1.2) {\scalebox{1.3}{$Y$}};
\node at (-1.3,0.1,-0.6) {\scalebox{1.3}{$X$}};
\coordinate (E1) at ($(H1)+2.5*(0.5,0.5,0)$);
\coordinate (E2) at ($(H1)+2.5*(0,0.5,0.5)$);
\coordinate (E3) at ($(H1)+2.5*(-0.5,0.5,0)$);
\coordinate (E4) at ($(H1)+2.5*(0,0.5,-0.5)$);
\draw(A)--(E1);
\draw [fill=red!70, fill opacity=0.5] (A)--(E1)--(E2)--(A);
\draw [dotted] (A)--(E4);
\draw [dotted] (E3)--(E4);
\draw[fill=red!70, fill opacity=0.5] (A)--(E3)--(E2)--(A);
\draw[fill=red!50, fill opacity=0.5] (H)--(E1)--(E2)--(H);
\draw [dotted] (H)--(E4);
\draw [dotted]  (E1)--(E4);
\draw[fill=red!70, fill opacity=0.5] (H)--(E3)--(E2)--(H);
\coordinate (S2) at ($0.6*(H)+2*(1,1,1)$);
\coordinate (S3) at ($0.6*(H)+2*(-1,1,1)$);
\coordinate (S4) at ($0.6*(H)+2*(-1,1,-1)$);
\coordinate (S1) at ($0.6*(H)+2*(1,1,-1)$);
\draw[fill=blue!80, fill opacity=0.3] (S1)--(S2)--(S3)--(S4)--(S1);
\draw[fill=black!60, fill opacity=0.4] (A)--(S1)--(S2)--(A); 
\draw[fill=red!0, fill opacity=0.] (A)--(S2)--(S3)--(A);
\draw (S1)--(A);
\draw (S2)--(A);
\draw (S3)--(A);
\draw[dotted] (S4)--(A); 
\shade[ball color=blue] (S1) circle (.1cm);
\shade[ball color=blue] (S2) circle (.1cm);
\shade[ball color=blue] (S3) circle (.1cm);
\shade[ball color=blue] (S4) circle (.1cm);
\shade[ball color=green1] (A) circle (.1cm);
\shade[ball color=red] (E1) circle (.1cm);
\shade[ball color=red] (E2) circle (.1cm);
\shade[ball color=red] (E3) circle (.1cm);
\shade[ball color=red] (E4) circle (.1cm);
\shade[ball color=green1] (H) circle (.1cm);
\node at ($0.6*(H)+2*(1.2,1.5,2)$) {\scalebox{1.3}{$\omega_1$}};
\node at ($0.6*(H)+2*(-1.2,1.2,1.2)$){\scalebox{1.3}{$\omega_2$}};
\node at ($0.6*(H)+2*(-1.2,1.1,-1.2)$){\scalebox{1.3}{$\omega_3$}};
\node at ($0.6*(H)+2*(1.0,1.3,-0.7)$) {\scalebox{1.3}{$\omega_4$}};
\node at ($(H1)+3*(0.6,0.4,0)$) {\scalebox{1.3}{$e_1$}};
\node at ($(H1)+3*(0.7,1.2,2.7)$) {\scalebox{1.3}{$e_2$}};
\node at ($(H1)+3*(-0.5,0.5,0)$) {\scalebox{1.3}{$e_3$}};
\node at ($(H1)+3*(0.1,0.4,-0.1)$) {\scalebox{1.3}{$e_4$}};
\node at ($(H)+(0.7,0.2,0)$) {\scalebox{1.3}{$e_5(u)$}};
\node at (-1.4,0.1,-2.4) {\scalebox{1.3}{$e_0$}};
\node at (-0.7,-1.2,-2.1) {\scalebox{1.2}{[$3$D-view]}};
\coordinate (M) at (-.8,5.5,7.5);
\coordinate (O) at ($(M)+(0,0,1)$);
\coordinate (S1) at ($(M)+1*(1,1,1)$);
\coordinate (S2) at ($(M)+1*(1,-1,1)$);
\coordinate (S3) at ($(M)+1*(-1,-1,1)$);
\coordinate (S4) at ($(M)+1*(-1,1,1)$);
\draw [fill=blue!20,opacity=.3](S1)--(S2)--(S3)--(S4)--cycle;
\draw[blue,thick](S1)--(S2)--(S3)--(S4)--cycle;
\coordinate (E1) at ($(M)+1*(0,1,1)$);
\coordinate (E2) at ($(M)+1*(1,0,1)$);
\coordinate (E3) at ($(M)+1*(0,-1,1)$);
\coordinate (E4) at ($(M)+1*(-1,0,1)$);
\draw[red,thick](E1)--(E2)--(E3)--(E4)--cycle;
\node at ($(M)+1*(1.1,1.3,1)$) {\scalebox{1.3}{$\textcolor{black}{\omega_1}$}};
\node at ($(M)+1*(-1.1,1.3,1)$) {\scalebox{1.3}{$\textcolor{black}{\omega_4}$}};
\node at ($(M)+1*(-0.9,-1.4,1)$) {\scalebox{1.3}{$\textcolor{black}{\omega_3}$}};
\node at ($(M)+1*(1.1,-1.4,1)$){\scalebox{1.3} {$\textcolor{black}{\omega_2}$}};
\draw[black!100,dashed][->]($(M)+(0,0,1)$)--($(M)+1*(1.6,0,1)$);
\draw[black!100,dashed][->]($(M)+(0,0,1)$)--($(M)+1*(0,1.6,1)$);
\node at ($(M)+(1.25,-0.65,0.3)$) {\scalebox{1.3}{$X$}};
\node at ($(M)+(0.40,2.2,2)$) {\scalebox{1.3}{$Y$}};
\node at ($(M)+(-0.3,1.25,1)$){\scalebox{1.3}{$e_1$}};
\node at ($(M)+(1.6,0.6,1.8)$){\scalebox{1.3}{$e_2$}};
\node at ($(M)+(0.5,-0.9,2)$) {\scalebox{1.3}{$e_3$}};
\node at ($(M)+(-1.2,0.0,1.3)$) {\scalebox{1.3}{$e_4$}};
\shade[ball color=blue] (S1) circle (.08cm);
\shade[ball color=blue] (S2) circle (.08cm);
\shade[ball color=blue] (S3) circle (.08cm);
\shade[ball color=blue] (S4) circle (.08cm);
\shade[ball color=red] (E1) circle (.08cm);
\shade[ball color=red] (E2) circle (.08cm);
\shade[ball color=red] (E3) circle (.08cm);
\shade[ball color=red] (E4) circle (.08cm);
\shade[ball color=green1] (O) circle (.08cm);
\node at ($(M)+(0.4,0.1,2)$) {\scalebox{1.3}{$u$}};
\node at ($(M)+(0.5,-2.1,2)$) {\scalebox{1.2}{[Top-view]}};
\end{tikzpicture}  
\caption{(Color online) Square-bit model. Normalized state-space $\mathbf{S}_{\bm\square}$ is the convex hull of extreme states $\omega_1:=(1,1,1)^\intercal$, $\omega_2:=(1,-1,1)^\intercal$, $\omega_3:=(-1,-1,1)^\intercal$, and $\omega_4:=(-1,1,1)^\intercal$. Effect Space $\mathbf{E}_{\bm\square}$ is the convex hull of extreme effects $e_0:=(0,0,0)^\intercal$, $e_1:=\frac{1}{2}(0,1,1)^\intercal$, $e_2:=\frac{1}{2}(1,0,1)^\intercal$, $e_3:=\frac{1}{2}(0,-1,1)^\intercal$, $e_4:=\frac{1}{2}(-1,0,1)^\intercal$, and $e_5:=(0,0,1)^\intercal=u$. Unit effect $u$ determines the normalization: $u(\omega)=1,~\forall\omega\in\mathbf{S}_{\bm\square}$.}
\label{fig1}
\end{figure}
Equal mixture of the extreme states yields the completely mixed state $\omega_{m}=1/4\sum_{i=1}^4\omega_i$. In this vectorized representation, the outcome probability reads as $e(\omega)=e\cdot\omega$, the Euclidean dot product in $\mathbb{R}^3$. The state-space ${\bf S}_{\bm\square}$ is the simplest case of a general class of models, where the state-space ${\bf S}_{2k}$ is specified by the convex hull of its $2k$ extreme points, with $k\in\{2,3,\cdots\}$ \cite{Janotta2011}. Geometrically ${\bf S}_{2k}$'s look like the symmetric even sided gons.\\\\
{\it Birkhoff-Violation in Non-Quantum GPTs.--}\label{Non-QGPT}
A transformation $T\in\mathbf{T}$, often called a channel, maps normalized states to normalized states. The requirement that a channel should preserve convexity, i.e.,
\begin{align*}
T(\sum\alpha_k\omega_k)=\sum\alpha_kT(\omega_k),\forall\omega_k\in\mathbf{S},~\alpha_k\ge0,\&\sum_k\alpha_k=1   
\end{align*} 
along with the demand $T(\Vec{0})=\Vec{0}$ ensures $T$ to be a linear operator on the $\mathbb{V}$ wherein the state-space is embedded \cite{Barrett2007}. Accordingly, a map $T\in\mathbf{T}_{\bm\square}$ on square-bit can be represented as a $3\times 3$ matrix.
\begin{align*}
    T=\begin{pmatrix}r_1 & r_2 & r_3\\[-.1cm]s_1 & s_2 & s_3\\[-.1cm]t_1 & t_2 & t_3\end{pmatrix}~\text{where} ~~ T\in\mathbf{T}_{\bm\square}
\end{align*}
which for brevity will be denoted as $T\equiv[\Vec{r},\Vec{s},\Vec{t}]$. The normalization preserving demand is ensured by fixing $\Vec{t}=(0,0,1)$, and thus channels are specified by $[\Vec{r},\Vec{s}]$. Note that, if $T_1$ and $T_2$ are two allowed channels, then convex mixture is also an allowed channel, implying $\mathbf{T}_{\bm\square}$ being a convex set. Accordingly, $\mathbf{T}_{\bm\square}$ is characterized by its extreme points as identified next.
\begin{proposition}\label{prop1}
$\mathbf{T}_{\bm\square}$ forms a convex polytope embedded in $\mathbb{R}^6$ having $36$ extreme points as listed in Table \ref{tab1}.\vspace{-.2cm}   
\end{proposition}
\begin{proof}
    A transformation $T\in\mathbf{T}_{\bm\square}$ maps the normalized state-space to itself. It is natural to assume that the transformation preserves convexity, i.e., 
\begin{align}
T\left(\sum\alpha_k\omega_k\right)&=\sum\alpha_kT\left(\omega_k\right),\\
\mbox{where}~\omega_k\in\mathbf{S}_{\bm\square}&~\&~ \alpha_k\ge0~\mbox{with}~\sum_k\alpha_k=1.\nonumber  
\end{align}
Note that, normalization of a state is determined by the unit effect $u$, i.e., $u(\omega)=1$, ensuring form of a normalized state to be $(a,b,1)^\intercal$ with $a,b\in[-1,1]$. The normalized state-space $\mathbf{S}_{\bm\square}$ turns out to be the convex hull of its four extreme points  
\begin{align*}
\mathbf{S}_{\bm\square}\equiv\mbox{ConvHul}
\left\{\!\begin{aligned}
\omega_1&:=(1,1,1)^\intercal,~~~~~~~~\omega_2:=(1,-1,1)^\intercal,\\
\omega_3&:=(-1,-1,1)^\intercal,~~\omega_4:=(-1,1,1)^\intercal
\end{aligned}\right\}.
\end{align*}
One can also define the notion of sub-normalized state $\omega$ that satisfies $u(\omega)\le1$. The set of sub-normalized states forms a polytope $\mathbf{S}^{SN}_{\bm\square}$ having the five extreme points (see Fig.\ref{fig1}):
\begin{align*}
\mathbf{S}^{SN}_{\bm\square}\equiv\mbox{ConvHul}\left\{\!\begin{aligned}
\omega_1,\omega_2,\omega_3,\omega_4,\\
\omega_0:=(0,0,0)^\intercal
\end{aligned}\right\}.
\end{align*}
Along with convexity, the demand $T(\omega_0)=\omega_0$ ensures $T$ to be a linear operator on $\mathbb{R}^3$ \cite{Barrett2007}, and hence it can be expressed as $3\times 3$ matrix
\begin{align}
T\equiv\begin{pmatrix}r_1 & r_2 & r_3\\[-.1cm]s_1 & s_2 & s_3\\[-.1cm]t_1 & t_2 & t_3\end{pmatrix}\equiv[\vec r,\vec s,\vec t].
\end{align}
A transformation $T$ being normalization preserving demands
\begin{align}
T\begin{pmatrix}u\\v\\1\end{pmatrix}&=\begin{pmatrix}r_1u+r_2v+r_3\\s_1u+s_2v+s_3\\t_1u+t_2v+t_3=1\end{pmatrix};~\forall~u,v\in[-1,1];\nonumber\\
\Rightarrow t_1&=t_2=0~~~\&~~~t_3=1. 
\end{align}
Consequently, a normalization preserving transformation (a channel) $T$ is completely specified by the tuple $[\vec r,\vec s]$. 

The requirement that $T(\omega)\in \mathbf{S}_{\bm\square},~\forall~\omega\in\mathbf{S}_{\bm\square}$, can be ensured by checking its action on the extreme points $\{\omega_i\}_{i=1}^4$ of $\mathbf{S}_{\bm\square}$. Accordingly, an admissible transformation must satisfy the following set of inequalities:
\begin{subequations}
\begin{align}
-1\le(\pm r_1 \pm r_2 + r_3)&\le+1,\label{ineqa}\\
-1\le(\pm s_1 \pm s_2 + s_3)&\le+1.\label{ineqd}
\end{align}
\end{subequations}
The inequalities (\ref{ineqa}-\ref{ineqd}) define a convex polytope in $\mathbb{R}^6$ determining the set of admissible channels $\mathbf{T}_{\bm\square}$. The polytope $\mathbf{T}_{\bm\square}$ can also be characterized by its extreme points, that can be obtained efficiently as follows:
\begin{itemize}
\item[] {\bf Step 1:} Since $6$ linearly independent planes define a unique point in $\mathbb{R}^6$, an extreme point must satisfy at least $6$ facet equations among the $16$
\begin{subequations}
\begin{align}
(\pm r_1 \pm r_2 + r_3)&=\pm 1,\label{eqa}\\
(\pm s_1 \pm s_2 + s_3)&=\pm 1\label{eqd}.
\end{align}
\end{subequations}
Thus, we can choose any $6$ facets from (\ref{eqa}-\ref{eqd}) to check whether they uniquely define a point.
\item[] {\bf Step 2:} If {\bf No}, then those $6$ facets do not define an extreme point. If {\bf Yes}, then it is checked whether the obtained solution satisfies (\ref{ineqa}-\ref{ineqd}). If these inequalities are {\bf Not}  satisfied, then the solution lies outside the polytope and hence it is not an extreme point. {\bf Else}, an extreme point is obtained.  
\item[] {\bf Step 3:} Repeat Step 1 and 2 for all possible $6$ tuple of facets so that we find all the extreme points. Since there are $16$ such facets, we need to repeat Step 1 and 2 for $16\choose6$ combinations.
\end{itemize}
\begin{figure}[t!]
\centering
\begin{tikzpicture}[scale=1.3]
\coordinate (A1) at (0,0);
\coordinate (A2) at (2.4,0.8);
\coordinate (A3) at (4,0);
\coordinate (A4) at (1.6,-0.8);
\coordinate (B1) at (2,2);
\coordinate (B2) at (2,-2);
\draw[->,thin] (.5,1.5) -- (.5,2.3);
\draw[->,thin] (.5,1.5) -- (1.3,1.5);
\draw[->,thin] (.5,1.5) -- (0,1);
\draw[thick,dotted] (A1) -- (A2);
\draw[thick,dotted] (A3) -- (A2);
\draw[thick,dotted] (B1) -- (A2);
\draw[thick,dotted] (B2) -- (A2);
\draw[fill=blue!40, fill opacity=0.7] (A1)--(A4)--(B1)--(A1);
\draw[fill=blue!40, fill opacity=0.3] (A3)--(A4)--(B1)--(A3);
\draw[fill=purple!40, fill opacity=0.3] (A1)--(A4)--(B2)--(A1);
\draw[fill=purple!40, fill opacity=0.6] (A3)--(A4)--(B2)--(A3);
\shade[ball color=blue] (A1) circle (.07cm);
\shade[ball color=blue] (A2) circle (.07cm);
\shade[ball color=blue] (A3) circle (.07cm);
\shade[ball color=blue] (A4) circle (.07cm);
\shade[ball color=blue] (B1) circle (.07cm);
\shade[ball color=blue] (B2) circle (.07cm);
\node at (.65,2.2) {\scalebox{1.3}{$z$}};
\node at (1.2,1.65) {\scalebox{1.3}{$y$}};
\node at (-0.12,1.1) {\scalebox{1.3}{$x$}};
\node at ($(A1)+(-.2,0.2)$) {\scalebox{1.3}{$\gamma^-_2$}};
\node at ($(A3)+(.2,0.2)$) {\scalebox{1.3}{$\gamma^+_2$}};
\node at ($(A2)+(.2,0.2)$) {\scalebox{1.3}{$\gamma^-_1$}};
\node at ($(A4)+(-.2,0.4)$) {\scalebox{1.3}{$\gamma^+_1$}};
\node at ($(B1)+(.3,0.1)$) {\scalebox{1.3}{$\gamma^+_3$}};
\node at ($(B2)+(.3,.1)$) {\scalebox{1.3}{$\gamma^-_3$}};
\end{tikzpicture}   
\caption{(Color online) The facets (\ref{eqa}) [similarly (\ref{eqd})] describe one of the {\it Platonic solids} Octahedron, which has the extreme points $\gamma^\pm_1:=(\pm1,0,0)^\intercal$, $\gamma^\pm_2:=(0,\pm1,0)^\intercal$, and $\gamma^\pm_3:=(0,0,\pm1)^\intercal$.}
\vspace{-.5cm}
\label{fig2}
\end{figure}
{\bf Remark 1:} Utilizing symmetry in (\ref{eqa}-\ref{eqd}), the extreme points can be obtained very easily. Eq.(\ref{eqa}) and Eq.(\ref{eqd}) respectively characterize the vector $\vec r=(r_1,r_2,r_3)^\intercal$ and $\vec s=(s_1,s_2,s_3)^\intercal$ independently. At-least $3$ facets from Eq.(\ref{eqa}), respectively from Eq.(\ref{eqd}), are necessary to uniquely define $\vec r$, respectively $\vec s$. As depicted in Fig.\ref{fig2}, the valid solutions can be obtained simply by plotting the respective facets. Accordingly, we have 
\begin{align}
\vec r,\vec s&\in\left\{\gamma_1^{\pm},\gamma_2^{\pm},\gamma_3^{\pm}\right\},\\
\gamma_1^{\pm}:=(\pm1,0,0)^\intercal,
\gamma_2^{\pm}&:=(0,\pm1,0)^\intercal,
\gamma_3^{\pm}:=(0,0,\pm1)^\intercal.\nonumber
\end{align}
Since each of the $\vec r$ and $\vec s$ take $6$ different values, we therefore have $36$ extreme channels (see Table \ref{tab1} and Remark 2 explaining the notation), i.e.  
\begin{align}\label{extch}
\mathbf{T}_{\bm\square}\equiv\mbox{ConvHul}\left\{T_{ext}:=\left[\gamma_i^{k_1},\gamma_j^{k_2}\right]\right\},\\
\mbox{with}~i,j\in\{1,2,3\}~\&~k_1,k_2\in\{+,-\}.\nonumber
\end{align}
~\vspace{-.5cm}\\
\end{proof}
The transformations described in Eq.(\ref{extch}) map the state space to itself. However, its unclear if these transformations map multipartite composite states to valid composite states. Within GPT framework various compositions are possible to describe multipartite system. Generally such compositions are constructed in accordance with the no-signaling principle and the local tomography demand \cite{Hardy2001}. Accordingly, the compositions can vary from the minimal tensor product, wherein entangled states are absent, to the maximal tensor product, wherein entangled effects are absent. 
\begin{lemma}\label{CPness}
For multipartite minimal and maximal tensor product compositions of a given GPT satisfying no restriction hypothesis, if $T$ maps the state space to itself then $\mathbb{I}^{\otimes n}\otimes T^{\otimes m}$ ($\mathbb{I}$ denotes the identity transformation) also maps the composite state space to itself ~$\forall n,m$. That is for minimal and maximal tensor product compositions positivity of map guarantees complete positivity of arbitrary tensor product of the map.
\end{lemma}
\begin{proof}
   Given a GPT $(\mathbf{S}, \mathbf{E}, \mathbf{T})$, let $T\in\mathbf{T} $ be a valid transformation of  $\mathbf{S}$ i.e. $T(\omega)\in \mathbf{S} ~\forall~ \omega$. Let  $\Omega$ be a valid state in either the minimal or maximal tensor product of $m+n$ copies of the given system. It is easy to see that in the case of minimal tensor product as $\Omega$ is separable $\mathbb{I}^{\otimes n}\otimes T^{\otimes m}(\Omega)$ is also separable and thus also belongs to the composite state space. The same can be argued for maximal tensor product also as follows. Assuming the \emph{no restriction hypothesis}, the  dual transformation $T^{\star} : \mathbf{E} \to \mathbf{E}$ can be defined as the unique map satisfying $e(T(\omega))=T^{\star}(e)(\omega)~ \forall ~\omega , e$ \cite{Plvala2023}. Now consider arbitrary effects $e_1,e_2,\cdots,e_{n+m}\in \mathbf{E}$ we have 
   \begin{align*}
       e_1\otimes e_2\otimes\cdots\otimes e_{n+m}\left(\mathbb{I}^{\otimes n}\otimes T^{\otimes m}(\Omega)\right)=\\e_1\otimes\cdots \otimes e_n\otimes T^{\star}(e_{n+1})\otimes\cdots\otimes T^{\star}(e_{n+m})\left(\Omega\right)
   \end{align*}
   But from the definition of dual map as $T^{\star}$ maps $ \mathbf{E}$ to itself we have $T^{\star}(e_{i})\in\mathbf{E}~\forall~i\in\{n+1,\cdots,n+m\}$. Thus $e_1\otimes\cdots \otimes e_n\otimes T^{\star}(e_{n+1})\otimes\cdots\otimes T^{\star}(e_{n+m})$ is also a product effect and the right hand side of the above equation is always positive $\forall~e_1,e_2,\cdots,e_{n+m}\in \mathbf{E}$. This proves that $\mathbb{I}^{\otimes n}\otimes T^{\otimes m}(\Omega)$ in fact is a valid state in the maximal tensor product composition.
 \end{proof}
From this point onward, whenever we refer to a composite system, we will assume that the box composition is equipped with either the minimal or the maximal tensor product.
\begin{table}[h!]
\centering
\begin{tabular}{|c||c|c|c|c|c|c|}
\hline
&&&&&&\\[-.3cm]
$\Vec{r}|\Vec{s}$&$\gamma_1^+$&$\gamma_1^-$&$\gamma_2^+$&$\gamma_2^-$&$\gamma_3^+$&$\gamma_3^-$\\[.1cm]
\hline\hline
&&&&&&\\[-.3cm]
$\gamma_1^+$&~~$\cellcolor{green!25}_{34}^{}\bm\square^{12}_{}$~~&~~$\cellcolor{green!25}_{}^{34}\bm\square^{}_{12}$~~&~~$\cellcolor{blue!15}_{3}^{4}\bm\square^{1}_{2}$~~&~~$\cellcolor{purple!25}_{4}^{3}\bm\square^{2}_{1}$~~&~~$_{}^{34}\bm\square^{12}_{}$~~&~~$_{34}^{}\bm\square^{}_{12}$~~\\[.1cm]
\hline
&&&&&&\\[-.3cm]
$\gamma_1^-$&~~$\cellcolor{green!25}_{}^{12}\bm\square^{}_{34}$~~&~~$\cellcolor{green!25}_{12}^{}\bm\square^{34}_{}$~~&~~$\cellcolor{purple!25}_{2}^{1}\bm\square^{4}_{3}$~~&~~$\cellcolor{blue!15}_{1}^{2}\bm\square^{3}_{4}$~~&~~$_{}^{12}\bm\square^{34}_{}$~~&~~$_{12}^{}\bm\square^{}_{34}$~~\\[.1cm]
\hline
&&&&&&\\[-.3cm]
$\gamma_2^+$&~~$\cellcolor{purple!25}_{3}^{2}\bm\square^{1}_{4}$~~&~~$\cellcolor{blue!15}_{2}^{3}\bm\square^{4}_{1}$~~&~~$\cellcolor{green!25}_{23}^{}\bm\square^{14}_{}$~~&~~$\cellcolor{green!25}_{}^{23}\bm\square^{}_{14}$~~&~~$_{}^{23}\bm\square^{14}_{}$~~&~~$_{23}^{}\bm\square^{}_{14}$~~\\[.1cm]
\hline
&&&&&&\\[-.3cm]
$\gamma_2^-$&~~$\cellcolor{blue!15}_{4}^{1}\bm\square^{2}_{3}$~~&~~$\cellcolor{purple!25}_{1}^{4}\bm\square^{3}_{2}$~~&~~$\cellcolor{green!25}_{}^{14}\bm\square^{}_{23}$~~&~~$\cellcolor{green!25}_{14}^{}\bm\square^{23}_{}$~~&~~$_{}^{14}\bm\square^{23}_{}$~~&~~$_{14}^{}\bm\square^{}_{23}$~~\\[.1cm]
\hline
&&&&&&\\[-.3cm]
$\gamma_3^+$&~~$_{}^{}\bm\square^{12}_{34}$~~&~~$_{}^{}\bm\square^{34}_{12}$~~&~~$_{}^{}\bm\square^{14}_{23}$~~&~~$_{}^{}\bm\square^{23}_{14}$~~&~~$_{}^{}\bm\square^{1234}_{}$~~&~~$_{}^{}\bm\square^{}_{1234}$~~\\[.1cm]
\hline
&&&&&&\\[-.3cm]
$\gamma_3^-$&~~$_{34}^{12}\bm\square^{}_{}$~~&~~$_{12}^{34}\bm\square^{}_{}$~~&~~$_{23}^{14}\bm\square^{}_{}$~~&~~$_{14}^{23}\bm\square^{}_{}$~~&~~$_{}^{1234}\bm\square^{}_{}$~~&~~$_{1234}^{}\bm\square^{}_{}$~~\\[.1cm]
\hline
\end{tabular}
\caption{(Color online) Extreme channels acting on ${\bf S}_{\bm\square}$. Here, $\gamma_1^{\pm}:=(\pm1,0,0),~\gamma_2^{\pm}:=(0,\pm1,0),~ \gamma_3^{\pm}:=(0,0,\pm1)$. The set $\mathrm{RaRe}_{\boldsymbol{\square}}$ consists of the eight reversible transformations of the square-bit, namely the four rotations (shown in violet) and the four reflections (shown in pink) listed in Table~\ref{tab1}. In contrast, the set $\mathrm{BiSto}_{\boldsymbol{\square}}$ contains these eight reversible extremal maps together with eight additional extremal bistochastic maps, highlighted in green.
}\label{tab1}
\vspace{-.5cm}
\end{table}\\
{\bf Remark 2:} For short-hand, here we introduce a box notation for representing transformations acting on the states of the square-bit. An extremal transformation is completely specified by specifying its action on the extremal states. Staring from the initial configuration $ {}_{3}^{4}\square_{2}^{1}$, an extremal transformation maps extremal states $\left\{\omega_i\right\}_{i=1}^4$ to extremal states. Accordingly,  the transformation is specified by identifying this mapping, and thus can be generally represented by the symbol ${}_{\&}^{\#}\square_{\%}^{\star}$, where the labels at the four corners of the box denotes the images of the extremal states after the action of the transformation. The labels in general takes values from the power set of $[4]:=\{1,2,3,4\}$. For a better appreciation we discuss few examples below:
\begin{itemize}
\item ${}_{4}^{3}\square_{1}^{2}$ denotes the transformation $T$ whose action on extreme states are given by $T(\omega_1)=\omega_2,~T(\omega_2)=\omega_1,~T(\omega_3)=\omega_4,~T(\omega_4)=\omega_3$. 
\item ${}_{34}^{}\square_{}^{12}$ denotes the transformation $T'$ whose action on extreme states are given by $T'(\omega_1)=\omega_1,~T'(\omega_2)=\omega_1,~T'(\omega_3)=\omega_3,~T'(\omega_4)=\omega_3$.
\item ${}_{}^{}\square_{}^{1234}$ denotes the transformation $T^{''}$ whose action on extreme states are given by $T^{''}(\omega_i)=\omega_1,\forall~i\in\{1,2,3,4\}$.
\end{itemize}
This notation therefore provides a compact graphical representation of extremal channels on ${\bf S}_{\square}$.
As already discussed, an extremal transformation can also be parametrized as $T_{\mathrm{ext}} \equiv [\vec r,\vec s]$, with 
\begin{align*}
\vec r,\vec s \in \Big\{\gamma_1^{\pm}:=(\pm1,0,0),~\gamma_2^{\pm}:=(0,\pm1,0),~\gamma_3^{\pm}:=(0,0,\pm1)\Big\}.
\end{align*}
Manifestly, we have 
\begin{align*}
{}_{4}^{3}\square_{1}^{2}\equiv T[\gamma^+_1,\gamma^-_2],~~{}_{34}^{}\square_{}^{12}\equiv T[\gamma^+_1,\gamma^+_1],~~ {}_{}^{}\square_{}^{1234}\equiv T[\gamma^+_3,\gamma^+_3].
\end{align*}
Throughout, we will use both these notations interchangeably.
 The set of channels obtained through stochastic application of reversible channels are called random reversible channels which is denoted as $\mathrm{RaRe}_{\bm\square}$. It has been argued in \cite{DallArno2017} that action of these reversible channels on part of a composite system always lead to bona fide probability. The remaining being analogous to quantum measure-and-prepare channels are entanglement-breaking and hence they also satisfy the required demand. A bistochastic channel keeps completely mixed state invariant, and hence demands $r_3=0=s_3$; $\mathrm{BiSto}_{\bm\square}$ denotes the set of all such channels. Our next result analyses the status of Birkhoff theorem in square-bit. 
\begin{theorem}\label{theo1}
$\mathrm{RaRe}_{\bm\square}\subsetneq\mathrm{BiSto}_{\bm\square}$, and thus Birkhoff theorem does not hold true in Square-bit model.
\end{theorem}
\begin{proof}
Let $\widetilde{T}$ be one of the eight maps (shown in green entries) in Table \ref{tab1}. We now argue that $\widetilde{T}\in\mathrm{BiSto}_{\bm\square}$ but 
$\widetilde{T}\notin\mathrm{RaRe}_{\bm\square}$.
As $\widetilde{T}$ is an extreme transformation, , if it were an element of $\mathrm{RaRe}_{\boldsymbol{\square}}$, then it would necessarily be a reversible map. Thus to show that $\widetilde{T}\notin \mathrm{RaRe}_{\boldsymbol{\square}}$, it suffices to prove that $\widetilde{T}$ is not injective. However, from the explicit action of $\widetilde{T}$, there exist two distinct pure extremal states $\omega_i\neq\omega_j$ such that $\widetilde{T}(\omega_i)=\widetilde{T}(\omega_j).$
Thus the set $\mathrm{RaRe}_{\boldsymbol{\square}}$ is the convex hull of the eight reversible transformations of the square-bit, namely the four rotations (violet color shade) and four reflections (pink color shade) listed in Table \ref{tab1}. In contrast, the set $\mathrm{BiSto}_{\boldsymbol{\square}}$ contains these eight reversible extremal maps together with eight additional extremal bistochastic maps, shown by the green entries in Table \ref{tab1}.
\end{proof}
{\it Characterization of Mixed unitary channels and unital channels.--}\label{mu}\\
{\bf Remark 3:} 
Having identified the extremal transformations of the square-bit theory, we now characterize the reversible, Random Reversible (RaRe), and unital transformations.Only $8$ , among the $36$ extreme channels in Eq.(\ref{extch}), are reversible transformations on $\mathbf{S}_{\bm\square}$. Accordingly, $\mathrm{RaRe}_{\bm\square}$ forms a sub-polytope of ${\bf T}_{\bm\square}$, and is given by
\begin{align}\label{raext}
\mathrm{RaRe}_{\bm\square}\equiv\mbox{ConvHul}\left\{T_{ext}:=\left[\gamma^{k_1}_i,\gamma^{k_2}_j\right]\right\},\\
\mbox{with}~i,j\in\{1,2\}~\&~i\neq j;~k_1,k_2\in\{+,-\}.\nonumber   
\end{align}
As pointed out in Theorem \ref{theo1}, $\mathrm{RaRe}_{\bm\square}\subsetneq\mathrm{BiSto}_{\bm\square}\subsetneq{\bf T}_{\bm\square}$, with    
\begin{align}
\mathrm{BiSto}_{\bm\square}\equiv\mbox{ConvHul}\left\{T_{ext}:=\left[\gamma^{k_1}_i,\gamma^{k_2}_j\right]\right\},\\
\mbox{with}~i,j\in\{1,2\};~k_1,k_2\in\{+,-\}.\nonumber
\end{align}
~\vspace{-.5cm}\\
This enables us to characterize reversible, RaRe, and bistochastic transformations through their respective criteria.\\\\
{\it State transformation criteria and monotones.--}\label{state_transfomration}
Given two vectors $\vec x, \vec y\in\mathbb{R}^n$, $\vec x$ is said to majorize $\vec y$, denoted as $\vec x\succ \vec y$, when $\sum_{i=1}^kx_i^{\downarrow}\ge \sum_{i=1}^ky_i^{\downarrow},~\forall~k\in\{1,\cdot\cdot,n\}$ with equality holding for $k=n$; here $x_i^{\downarrow}$ denotes $i^{th}$ largest entry of $\vec x$ \cite{Marshall2011}. Birkhoff theorem ensures that a classical state (probability vector)  $\Vec{p}$ can be converted into another state $\Vec{q}$ under a bistochastic operation {\it if and only if} the former majorizes the later, i.e., $\Vec{p}\succ\Vec{q}$ \cite{Birkhoff1946}. Similarly, a quantum state $\rho$ can be mapped to another state $\sigma$ under a RaRe operation {\it if and only if} $\vec \lambda_{\rho}\succ\vec \lambda_{\sigma}$, where $\vec \lambda_{\star}$ denotes spectral of the corresponding state \cite{Horodecki2003(2)}. Notably, the state transformation conditions remain to be the same under the set of noisy operations as well as under the set of unital operations, although for qutrit and beyond these three sets are distinct \cite{Chiribella2017}. Coming back to the square-bit model, we now completely specify the state transformation criteria.

Note that, a generic state $\omega\in{\bf S}_{\bm\square}$ reads as $\omega\equiv(p,1-p|q,1-q)^\intercal\equiv(2p-1,2q-1,1)^\intercal$, with $p,q\in[0,1]$. As shown in the Top-view (see Fig.\ref{fig3}), plotting $p$ and $q$ along $x$ and $y$, respectively we have:
\begin{itemize}
\item[(1)] $p=q=1/2$: represent completely mixed state $\omega_m$;
\item[(2)] $p,q\in(1/2,1]$: states are in $1^{st}$ quadrant (w.r.t. $\omega_m$); 
\item[(3)] $p\in[0,1/2)~\&~q\in(1/2,1]$: states are in $2^{nd}$ quadrant;
\item[(4)] $p,q\in[0,1/2)$: states are in $3^{rd}$ quadrant;
\item[(5)] $p\in(1/2,1]~\&~q\in[0,1/2)$: states are in $4^{th}$ quadrant;
\end{itemize}
Any state from $2^{nd}$, $3^{rd}$, and $4^{th}$ quadrant can be brought into $1^{st}$ quadrant by applying suitable reversible operations. Therefore, without loss of any generality we can limit our study to the pairs $\omega\equiv(p,1-p|q,1-q)^\intercal~\&~ \omega^\prime\equiv(p^\prime,1-p^\prime|q^\prime,1-q^\prime)^\intercal$, with $p,q,p^\prime,q^\prime\in[1/2,1]$.
\begin{theorem}\label{theo2}
Given $\omega,\omega^\prime\in{\bf S}_{\bm\square}$, the transformation $\omega\to\omega^\prime$ is possible under $\mathrm{BiSto}_{\bm\square}$ if and only if, $\max\{p,q\}\ge \max\{p^\prime,q^\prime\}$.
\begin{proof}
As already mentioned we can consider an arbitrary initial state $\omega=(p, 1-p | q, 1-q)^\intercal$, with $p,q\in[1/2,1]$. Two possibilities can arise: (i) $q\ge p\ge1/2$, and (ii) $p\ge q\ge1/2$.  

For case (i), achievable set of states under $\mathrm{BiSto}_{\bm\square}$ forms a convex polytope ${\bf S}^\omega_\mathrm{BiSto}\subset{\bf S}_{\bm\square}$, whose extreme points are given by (see Fig.\ref{fig3})
\begin{figure}[h!]
\centering
\begin{tikzpicture}[scale=1.8]
\draw[dashed] (-1.8,0,0)--(1.8,0,0); 
\draw[->,thin] (1.8,-1.8,0)--(2.2,-1.8,0);
\draw[dashed] (0,-1.8,0)--(0,1.8,0);
\draw[->,thin] (-1.8,1.8,0)--(-1.8,2.2,0);
\coordinate (S1) at (1.8,1.8,0);
\coordinate (S2) at (-1.8,1.8,0);
\coordinate (S3) at (-1.8,-1.8,0);
\coordinate (S4) at (1.8,-1.8,0);
\coordinate (A1) at (.3,1.2,0);
\coordinate (A2) at (-.3,1.2,0);
\coordinate (A3) at (-1.2,.3,0);
\coordinate (A4) at (-1.2,-.3,0);
\coordinate (A5) at (-.3,-1.2,0);
\coordinate (A6) at (.3,-1.2,0);
\coordinate (A7) at (1.2,-.3,0);
\coordinate (A8) at (1.2,.3,0);
\coordinate (B1) at (.75,.75,0);
\coordinate (B2) at (-.75,.75,0);
\coordinate (B3) at (-.75,-.75,0);
\coordinate (B4) at (.75,-.75,0);
\coordinate (C1) at (1.2,1.2,0);
\coordinate (C2) at (-1.2,1.2,0);
\coordinate (C3) at (-1.2,-1.2,0);
\coordinate (C4) at (1.2,-1.2,0);
\draw [thick] (S1)--(S2)--(S3)--(S4)--(S1);
\draw [dotted,thick,fill=blue!50, fill opacity=0.3] (C1)--(C2)--(C3)--(C4)--(C1);
\shade[ball color=black] (S1) circle (.05cm);
\shade[ball color=black] (S2) circle (.05cm);
\shade[ball color=black] (S3) circle (.05cm);
\shade[ball color=black] (S4) circle (.05cm);
\shade[ball color=abc] (A1) circle (.05cm);
\shade[ball color=blue] (C1) circle (.05cm);
\shade[ball color=blue] (C2) circle (.05cm);
\shade[ball color=blue] (C3) circle (.05cm);
\shade[ball color=blue] (C4) circle (.05cm);
\shade[ball color=green1] (0,0,0) circle (.05cm);
\node at ($(A1)+(0,0.2,0)$) {\scalebox{1.5}{$\omega$}};
\node at ($(C1)+(.1,0.25,0)$) {\scalebox{1.5}{$\omega^\prime_{\scaleto{00}{3pt}}$}};
\node at ($(C2)+(0,0.25,0)$) {\scalebox{1.5}{$\omega^\prime_{\scaleto{01}{3pt}}$}};
\node at ($(C3)+(0,-0.25,0)$) {\scalebox{1.5}{$\omega^\prime_{\scaleto{11}{3pt}}$}};
\node at ($(C4)+(.1,-0.25,0)$) {\scalebox{1.5}{$\omega^\prime_{\scaleto{10}{3pt}}$}};
\node at ($(S1)+(.1,0.25,0)$) {\scalebox{1.5}{$\omega_{\scaleto{1}{3pt}}$}};
\node at ($(S2)+(-0.2,0.2,0)$) {\scalebox{1.5}{$\omega_{\scaleto{4}{3pt}}$}};
\node at ($(S3)+(0,-0.25,0)$) {\scalebox{1.5}{$\omega_{\scaleto{3}{3pt}}$}};
\node at ($(S4)+(.1,-0.25,0)$) {\scalebox{1.5}{$\omega_{\scaleto{2}{3pt}}$}};
\node at (.22,0.15,0) {\scalebox{1.5}{$\omega_{\scaleto{m}{2pt}}$}};
\node at (2.1,-1.6,0) {\scalebox{1.5}{$p$}};
\node at (-1.65,2.1,0) {\scalebox{1.5}{$q$}};
\node at ($(C1)+(0.3,-0.45,0)$) {\scalebox{1.5}{\textcolor{blue}{$1^{st}$}}};
\node at ($(C2)+(-0.3,-0.45,0)$) {\scalebox{1.5}{\textcolor{blue}{$2^{nd}$}}};
\node at ($(C3)+(-0.3,0.45,0)$) {\scalebox{1.5}{\textcolor{blue}{$3^{rd}$}}};
\node at ($(C4)+(0.3,0.45,0)$) {\scalebox{1.25}{\textcolor{blue}{$4^{th}$}}};
\end{tikzpicture} 
\caption{(Color online) Starting with the state $\omega=(p,1-p|q,1-q)^\intercal$, all the states within the shaded region can be reached under $\mathrm{BiSto}_{\bm\square}$.}
\vspace{-.5cm}
\label{fig3}
\end{figure}
\begin{align*}
\left\{\!\begin{aligned}
\omega^\prime_{00}:=(q, 1-q | q, 1-q)^\intercal,~\omega^\prime_{01}:=(1-q,q | q, 1-q)^\intercal\\
\omega^\prime_{10}:=(q, 1-q | 1-q,q)^\intercal,~\omega^\prime_{11}:=(1-q,q |1-q,q)^\intercal\end{aligned}\right\}. 
\end{align*}
Therefore, any point $\omega^\prime$ lying within the set ${\bf S}^\omega_\mathrm{BiSto}\equiv\mbox{ConvHull}\{\omega^\prime_{00},\omega^\prime_{01},\omega^\prime_{10},\omega^\prime_{11}\}$ can be achieved under $\mathrm{BiSto}_{\bm\square}$ while starting with the state $\omega$. Once again, $\omega^\prime$ can be brought into the form $\omega^\prime=(p^\prime, 1-p^\prime | q^\prime, 1-q^\prime)^\intercal$, with $p^\prime,q^\prime\in[1/2,1]$. Therefore, to lie within the allowed region the following inequalities appear as necessary and sufficient criteria:  
\begin{align}
\max\{p^\prime,q^\prime\}\le q. 
\end{align}
Similarly, analyzing case (ii) we have,
$\max\{p^\prime,q^\prime\}\le p$.
Therefore, in general the state transformation criteria ($\omega\to\omega^\prime$) under $\mathrm{BiSto}_{\bm\square}$ reads as \begin{align}\label{eqbi}\max\{p^\prime,q^\prime\}\le \max\{p,q\}. \end{align}This completes the proof.
\end{proof}
\end{theorem}
\begin{theorem}\label{theo3}
Given $\omega,\omega^\prime\in{\bf S}_{\bm\square}$, the transformation $\omega\to\omega^\prime$ is possible under $\mathrm{RaRe}_{\bm\square}$ if and only if, $\max\{p,q\}\ge \max\{p^\prime,q^\prime\}$ and $p+q\ge p^\prime+q^\prime$.
\end{theorem}
\begin{proof}
Starting with the state $\omega=(p, 1-p | q, 1-q)^\intercal$, with $q\ge p\ge1/2$, the set of states achievable under $\mathrm{RaRe}_{\bm\square}$ forms a convex polytope ${\bf S}^\omega_\mathrm{RaRe}\subset{\bf S}_{\bm\square}$, whose extreme points are given by (see Fig.\ref{fig4})
\begin{figure}[t!]
\centering
\begin{tikzpicture}[scale=1.7]
\draw[dashed] (-1.8,0,0)--(1.8,0,0); 
\draw[->,thin] (1.8,-1.8,0)--(2.2,-1.8,0);
\draw[dashed] (0,-1.8,0)--(0,1.8,0);
\draw[->,thin] (-1.8,1.8,0)--(-1.8,2.2,0);
\coordinate (S1) at (1.8,1.8,0);
\coordinate (S2) at (-1.8,1.8,0);
\coordinate (S3) at (-1.8,-1.8,0);
\coordinate (S4) at (1.8,-1.8,0);
\coordinate (A1) at (.3,1.2,0);
\coordinate (A2) at (-.3,1.2,0);
\coordinate (A3) at (-1.2,.3,0);
\coordinate (A4) at (-1.2,-.3,0);
\coordinate (A5) at (-.3,-1.2,0);
\coordinate (A6) at (.3,-1.2,0);
\coordinate (A7) at (1.2,-.3,0);
\coordinate (A8) at (1.2,.3,0);
\draw [thick] (S1)--(S2)--(S3)--(S4)--(S1);
\draw [thick,dashed,red] (-.8,1.2,0)--(.8,1.2,0);
\draw [thick,dashed,red] (1.7,-.2,0)--(-.2,1.7,0);
\draw [thick,dashed,red] (1.2,-.8,0)--(1.2,.8,0);
\draw [thick,dotted,fill=purple!50, fill opacity=0.3] (A1)--(A2)--(A3)--(A4)--(A5)--(A6)--(A7)--(A8)--(A1);
\shade[ball color=black] (S1) circle (.05cm);
\shade[ball color=black] (S2) circle (.05cm);
\shade[ball color=black] (S3) circle (.05cm);
\shade[ball color=black] (S4) circle (.05cm);
\shade[ball color=abc] (A1) circle (.05cm);
\shade[ball color=blue] (A2) circle (.05cm);
\shade[ball color=blue] (A3) circle (.05cm);
\shade[ball color=blue] (A4) circle (.05cm);
\shade[ball color=blue] (A5) circle (.05cm);
\shade[ball color=blue] (A6) circle (.05cm);
\shade[ball color=blue] (A7) circle (.05cm);
\shade[ball color=blue] (A8) circle (.05cm);
\shade[ball color=green1] (0,0,0) circle (.05cm);
\node at ($(A1)+(0.3,0.2,0)$) {\scalebox{1.5}{$\omega^\prime_1=\omega$}};
\node at ($(A8)+(0.2,0.2,0)$) {\scalebox{1.5}{$\omega^\prime_2$}};
\node at ($(A7)+(0.2,-0.1,0)$) {\scalebox{1.5}{$\omega^\prime_3$}};
\node at ($(A6)+(0.1,-0.1,0)$) {\scalebox{1.5}{$\omega^\prime_4$}};
\node at ($(A5)+(-0.15,-0.1,0)$) {\scalebox{1.5}{$\omega^\prime_5$}};
\node at ($(A4)+(-0.2,-0.1,0)$) {\scalebox{1.5}{$\omega^\prime_6$}};
\node at ($(A3)+(-0.15,+0.2,0)$) {\scalebox{1.5}{$\omega^\prime_7$}};
\node at ($(A2)+(-0.15,+0.2,0)$) {\scalebox{1.5}{$\omega^\prime_8$}};
\node at ($(S1)+(.1,0.25,0)$) {\scalebox{1.5}{$\omega_{\scaleto{1}{3pt}}$}};
\node at ($(S2)+(-0.2,0.2,0)$) {\scalebox{1.5}{$\omega_{\scaleto{4}{3pt}}$}};
\node at ($(S3)+(0,-0.25,0)$) {\scalebox{1.5}{$\omega_{\scaleto{3}{3pt}}$}};
\node at ($(S4)+(.1,-0.25,0)$) {\scalebox{1.5}{$\omega_{\scaleto{2}{3pt}}$}};
\node at (.22,0.15,0) {\scalebox{1.5}{$\omega_{\scaleto{m}{2pt}}$}};
\node at (2.1,-1.6,0) {\scalebox{1.5}{$p$}};
\node at (-1.65,2.1,0) {\scalebox{1.5}{$q$}};
\node at ($(C1)+(0.3,-0.45,0)$) {\scalebox{1.5}{\textcolor{blue}{$1^{st}$}}};
\node at ($(C2)+(-0.3,-0.45,0)$) {\scalebox{1.5}{\textcolor{blue}{$2^{nd}$}}};
\node at ($(C3)+(-0.3,0.45,0)$) {\scalebox{1.5}{\textcolor{blue}{$3^{rd}$}}};
\node at ($(C4)+(0.3,0.45,0)$) {\scalebox{1.25}{\textcolor{blue}{$4^{th}$}}};
\end{tikzpicture} 
\caption{(Color online) Starting with the state $\omega=(p,1-p|q,1-q)^\intercal$, all the states within the shaded region can be reached under $\mathrm{RaRe}_{\bm\square}$.}
\vspace{-.5cm}
\label{fig4}
\end{figure}
\begin{align*}
\left\{\!\begin{aligned}
\omega^\prime_1:=(p, 1-p | q, 1-q)^\intercal,~\omega^\prime_2:=(q,1-q | p, 1-p)^\intercal\\
\omega^\prime_3:=(q, 1-q | 1-p,p)^\intercal,~\omega^\prime_4:=(p,1-p |1-q,q)^\intercal\\
\omega^\prime_5:=(1-p, p | 1-q,q)^\intercal,~\omega^\prime_6:=(1-q,q |1-p,p)^\intercal\\
\omega^\prime_7:=(1-q, q | p,1-p)^\intercal,~\omega^\prime_8:=(1-p,p |q,1-q)^\intercal
\end{aligned}\right\}. 
\end{align*}
Therefore, any point $\omega^\prime$ lying within the set ${\bf S}^\omega_\mathrm{RaRe}$ can be achieved under $\mathrm{RaRe}_{\bm\square}$ while starting with the state $\omega$. Once again, $\omega^\prime$ can be brought into the form $\omega^\prime=(p^\prime, 1-p^\prime | q^\prime, 1-q^\prime)^\intercal$, with $p^\prime,q^\prime\in[1/2,1]$. Therefore, to lie within the allowed region the following inequalities appear as necessary and sufficient criteria:
\begin{align}
\max\{p^\prime,q^\prime\}\le q,\\
p^\prime+q^\prime\le p+q.
\end{align}
Similarly, analyzing the other case ($p\ge q\ge 1/2$), the general state transformation criteria ($\omega\to\omega^\prime$) under $\mathrm{RaRe}_{{\bm\square}}$ read as
\begin{subequations}
\begin{align}
\max\{p^\prime,q^\prime\}&\le \max\{p,q\},\label{eqra1}\\
p^\prime+q^\prime&\le p+q.\label{eqra2}
\end{align}
\end{subequations}
This completes the proof.
\end{proof} As it turns out, some transformations are not admissible under $\mathrm{RaRe}_{\bm\square}$, but otherwise possible under $\mathrm{BiSto}_{\bm\square}$, and this fact underlies the operational distinction as reported previously. At this point, we leave open the question of what other operational implications might arise due to these different state transformation criteria. 
{\it Entropic Monotones.--}\label{entropy}
Once the state transformation criteria are established, it becomes possible to identify `monotones' that hold operational interest. For an arbitrary GPT $(\mathbf{S},\mathbf{E},\mathbf{T})$, a function $f_{\mathbf{T}^\prime}:\mathbf{S}\mapsto\mathbb{R}$ is said to be a non-decreasing [non-increasing] monotone under $\mathbf{T}^\prime\subseteq\mathbf{T}$, if $f_{\mathbf{T}^\prime}(\omega)\leq f_{\mathbf{T}^\prime}(\omega^\prime)~[f_{\mathbf{T}^\prime}(\omega)\geq f_{\mathbf{T}^\prime}(\omega^\prime)]$ whenever $\omega\xrightarrow[]{T}\omega^\prime$ for some $T\in\mathbf{T}^\prime$. Such monotones are often useful for comparing resources in different states. For instance, Shannon entropy $H(\Vec{p}):=-\sum_{i}p_i\log p_i$ of a classical state $\Vec{p}\equiv\{p_i\}$ is a (non-decreasing) monotone under the set of bistochastic operations, which quantifies the degree of randomness in a state that cannot be decreased further under bistochastic operations \cite{Marshall2011}. Similarly, von Neumann entropy $S(\rho):=-\tr[\rho\log\rho]$ of a quantum state $\rho$ is a (non-decreasing) monotone under the set of RaRe operations, noisy operations, as well as unital
operations \cite{Chiribella2017}. Coming back to the square-bit theory, we recall that a Schur-convex function $f:\mathbb{R}^n\to \mathbb{R}$ satisfies $f(\vec x)\le f(\vec y)$ whenever $\vec x\succ \vec y$ \cite{Peajcariaac1992}. Accordingly, Theorems \ref{theo2} \& \ref{theo3} lead to the following entropic monotones:
\begin{itemize}
\item[I.] {\it Bi-stochastic operations}: Defining $\eta:=\max\{p,q\}$ and $\eta^\prime:=\max\{p^\prime,q^\prime\}$, Eq.(\ref{eqbi}) reads as
\begin{align}
\eta^\prime\le\eta.
\end{align}
Accordingly, $\vec{\eta^\prime}\prec\vec\eta$, where $\vec\xi:=(\xi,1-\xi)^\intercal$. Since Shannon entropy $H(x):=-x\log x-(1-x)\log(1-x)$ \cite{Shannon1948} is a Schur-concave function \cite{Peajcariaac1992}, therefore we have 
\begin{equation}
H(\vec{\eta}) \leq H(\vec{\eta^\prime}).
\end{equation}

\item[II.]{\it Random reversible operations}: Given the pair of states $\omega=(p,1-p|q,1-q)^\intercal$ ($p,q\ge1/2$) and $\omega^\prime=(p^\prime,1-p^\prime|q^\prime,1-q^\prime)^\intercal$ ($p^\prime,q^\prime\ge1/2$), where $\omega\to\omega^\prime$ under $\mathrm{RaRe}_{\bm\square}$, we can consider two probability vectors
\begin{subequations}
\begin{align}
\vec r&:=\frac{1}{2}(\eta,\delta,1-\delta,1-\eta)^\intercal,\\
\vec{r^\prime}&:=\frac{1}{2}(\eta^\prime,\delta^\prime,1-\delta^\prime,1-\eta^\prime)^\intercal,
\end{align}
\end{subequations}
where $\eta:=\max\{p,q\}~\&~\delta:=\min\{p,q\}$ and $\eta^\prime:=\max\{p^\prime,q^\prime\}~\&~\delta^\prime:=\min\{p^\prime,q^\prime\}$. The conditions (\ref{eqra1})-(\ref{eqra2}) lead to the majorization relation $\Vec{r^\prime}\prec\vec r$, and can be re-expressed as
\begin{subequations}
\begin{align}
\eta^\prime&\le\eta,\\
\eta^\prime+\delta^\prime&\le\eta+\delta.
\end{align}
Accordingly, we have the following two monotones
\begin{align}
H(\vec{\eta}) &\leq H(\vec{\eta^\prime}),\\
H(\vec{\eta})+H(\vec{\delta})&\leq H(\vec{\eta^\prime})+H(\vec{\delta^\prime})\nonumber\\
\implies H(\vec{p})+H(\vec{q})&\leq H(\vec{p^\prime})+H(\vec{q^\prime}). 
\end{align}
\end{subequations}
\end{itemize}
\vspace{0.2cm}
{\it Resource-Theoretic Implications in RAC.--}\label{RAC}
We explore the information-theoretic significance of this violation in Random Access Code (RAC) tasks, showing that unital operations outperform RaRe transformations. This operational gap highlights Birkhoff violation as a resource and reveals its role in enhancing communication protocols within sharp microcanonical theories. To this aim, we consider a variant of $2\mapsto1$ RAC task, a sender (Alice) is given a random bit string $x_0x_1\in\{0,1\}^{\times2}$, whereas a distant receiver (Bob) has to guess $x_y$ given the random input $y\in\{0,1\}$ \cite{Wiesner1983,Ambainis2002}. Communicating a qubit is known to provide an advantage over classical bit communication in several information-processing tasks. In contrast, within the operational framework considered here, the transmission of a square-bit yields an advantage over the optimal qubit-based strategy.
In $2\mapsto 1$ RAC task, Alice is provided a random bit string $x_0x_1\in\{0,1\}^{\times2}$, whereas the Bob has to guess $x_y$ given the random input $y\in\{0,1\}$. If only one classical bit of communication is allowed from Alice to Bob, then Bob can perfectly guess one of the bits, while the other bit is completely random, making the classical Average success probability $3/4$. Interestingly, with a qubit communication, the optimal Average success probability becomes $\frac{1}{2}\left(1+\frac{1}{\sqrt{2}}\right)$ \cite{Ambainis2008}. A qubit strategy yielding the optimal success reads as:
\begin{subequations}
\begin{align}
&\mbox{\it Encoding:}\nonumber\\
&\left\{\!\begin{aligned}
00\mapsto\ket{+},
01\mapsto\ket{0},
10\mapsto\ket{1},
11\mapsto\ket{-}\end{aligned}\right\};\\
&\mbox{\it Decoding:}\nonumber\\
&\left\{\!\begin{aligned}
x_0:~M_0\equiv\left\{\pi_0^i:=\frac{1}{2}\left(\mathbf{I}+(-1)^i\frac{\sigma_x+\sigma_z}{\sqrt{2}}\right)\right\}\\
x_1:~M_1\equiv\left\{\pi_1^i:=\frac{1}{2}\left(\mathbf{I}+(-1)^i\frac{\sigma_x-\sigma_z}{\sqrt{2}}\right)\right\}
\end{aligned}\right\},
\end{align}
\end{subequations}
where, $i\in\{0,1\}$ and $\ket{\pm}:=\frac{1}{\sqrt{2}}(\ket{0}\pm\ket{1})$. If Alice is allowed to communicate state from square-bit they end up with a perfect strategy \cite{Banik2015}: 
\begin{subequations}
\begin{align}
&\mbox{\it Encoding:}\nonumber\\
&\left\{\!\begin{aligned}
00\mapsto\omega_1,~01\mapsto\omega_4,~
10\mapsto\omega_2,~11\mapsto\omega_3\end{aligned}\right\};\\ 
&\mbox{\it Decoding:}\nonumber\\
&\left\{\!\begin{aligned}
x_0:~M_0\equiv M_{13}\equiv\left\{e_0^0:=e_1,~e_0^1:=e_3\right\}\\
x_1:~M_1\equiv M_{24}\equiv\left\{e_1^0:=e_2,~e_1^1:=e_4\right\}
\end{aligned}\right\}~.\label{boxd}
\end{align}
\end{subequations}
 Notably, the optimal strategies in qubit as well as in square-bit invoke pure states for encoding at the sender's end. Since preparation of such states is costly \cite{Szilard1929,Bennett1993,Maruyama2009}, here we consider that the sender does not possess any state preparation device, rather the referee provides her with some fixed state. Alice is allowed to apply any channel $T\in\mathbf{T}^\prime\subseteq\mathbf{T}_{\bm\square}$ on this state to encode the given string, and accordingly sends the encoded system to Bob. By $P_\omega[\mathbf{T}^\prime]$ we denote the optimal average success achieved in $2\mapsto1$ RAC when Alice is provided the state $\omega$, and she is allowed to apply any channel from the set $\mathbf{T}^\prime$ on this state for encoding. 
\begin{theorem}\label{theo4}
Given a state $\omega=(2p-1,2q-1,1)^\intercal\in\mathbf{S}_{\bm\square}$,$P_\omega[\mathbf{RaRe}_{\bm\square}]=\frac{1}{4}(2+|2p-1|+|2q-1|)\le\frac{1}{2}(1+\max\{|2p-1|,|2q-1|\})=P_\omega[\mathbf{BiSto}_{\bm\square}]$.  
\end{theorem}
\begin{proof}
Let the state $\tilde\omega=(2p-1,2q-1,1)^\intercal$ is provided to Alice by the referee, where $0\le p,q\le 1$. Applying reversible maps this state can be converted into $\omega=(2p-1,2q-1,1)^\intercal\in\mathbf{S}_{\bm\square}$, with $p,q\in[1/2,1]$. Therefore, without loss of any generality we can restrict our analysis to the parameter ranges $p,q\in[1/2,1]$.\\
{\it Optimal success under $\mathrm{RaRe}_{\bm\square}$.--} \\Alice's encoding ($x_0x_1\mapsto\omega_{x_ox_1}$) is given by (see Fig.\ref{fig5})
\begin{figure}[h!]\label{fig5}
\centering
\begin{tikzpicture}[scale=1.8]
\draw[dashed] (-1.8,0,0)--(1.8,0,0); 
\draw[->,thin] (1.8,-1.8,0)--(2.0,-1.8,0);
\draw[dashed] (0,-1.8,0)--(0,1.8,0);
\draw[->,thin] (-1.8,1.8,0)--(-1.8,2.0,0);
\coordinate (S1) at (1.8,1.8,0);
\coordinate (S2) at (-1.8,1.8,0);
\coordinate (S3) at (-1.8,-1.8,0);
\coordinate (S4) at (1.8,-1.8,0);
\coordinate (A1) at (.3,1.2,0);
\coordinate (A2) at (-.3,1.2,0);
\coordinate (A3) at (-1.2,.3,0);
\coordinate (A4) at (-1.2,-.3,0);
\coordinate (A5) at (-.3,-1.2,0);
\coordinate (A6) at (.3,-1.2,0);
\coordinate (A7) at (1.2,-.3,0);
\coordinate (A8) at (1.2,.3,0);
\coordinate (B1) at (.75,.75,0);
\coordinate (B2) at (-.75,.75,0);
\coordinate (B3) at (-.75,-.75,0);
\coordinate (B4) at (.75,-.75,0);
\coordinate (C1) at (1.2,1.2,0);
\coordinate (C2) at (-1.2,1.2,0);
\coordinate (C3) at (-1.2,-1.2,0);
\coordinate (C4) at (1.2,-1.2,0);
\draw [thick] (S1)--(S2)--(S3)--(S4)--(S1);
\draw [dotted,thick,fill=blue!50, fill opacity=0.3] (C1)--(C2)--(C3)--(C4)--(C1);
\draw [thin,fill=red!50, fill opacity=0.3] (A1)--(A2)--(A3)--(A4)--(A5)--(A6)--(A7)--(A8)--(A1);
\shade[ball color=black] (S1) circle (.05cm);
\shade[ball color=black] (S2) circle (.05cm);
\shade[ball color=black] (S3) circle (.05cm);
\shade[ball color=black] (S4) circle (.05cm);
\shade[ball color=abc] (A1) circle (.05cm);
\shade[ball color=blue] (A2) circle (.01cm);
\shade[ball color=blue] (A3) circle (.01cm);
\shade[ball color=blue] (A4) circle (.01cm);
\shade[ball color=blue] (A5) circle (.01cm);
\shade[ball color=blue] (A6) circle (.01cm);
\shade[ball color=blue] (A7) circle (.01cm);
\shade[ball color=blue] (A8) circle (.01cm);
\shade[ball color=red] (B1) circle (.05cm);
\shade[ball color=red] (B2) circle (.05cm);
\shade[ball color=red] (B3) circle (.05cm);
\shade[ball color=red] (B4) circle (.05cm);
\shade[ball color=blue] (C1) circle (.05cm);
\shade[ball color=blue] (C2) circle (.05cm);
\shade[ball color=blue] (C3) circle (.05cm);
\shade[ball color=blue] (C4) circle (.05cm);
\shade[ball color=green1] (0,0,0) circle (.05cm);
\node at ($(A1)+(0,0.2,0)$) {\scalebox{1.5}{$\omega$}};
\node at ($(C1)+(.1,0.25,0)$) {\scalebox{1.5}{$\omega^\prime_{\scaleto{00}{3pt}}$}};
\node at ($(C2)+(0,0.25,0)$) {\scalebox{1.5}{$\omega^\prime_{\scaleto{01}{3pt}}$}};
\node at ($(C3)+(0,-0.25,0)$) {\scalebox{1.5}{$\omega^\prime_{\scaleto{11}{3pt}}$}};
\node at ($(C4)+(.1,-0.25,0)$) {\scalebox{1.5}{$\omega^\prime_{\scaleto{10}{3pt}}$}};
\node at ($(B1)+(-.1,-0.25,0)$) {\scalebox{1.5}{$\omega_{\scaleto{00}{3pt}}$}};
\node at ($(B2)+(.2,-0.25,0)$) {\scalebox{1.5}{$\omega_{\scaleto{01}{3pt}}$}};
\node at ($(B3)+(.2,0.25,0)$) {\scalebox{1.5}{$\omega_{\scaleto{11}{3pt}}$}};
\node at ($(B4)+(-.1,0.25,0)$) {\scalebox{1.5}{$\omega_{\scaleto{10}{3pt}}$}};
\node at ($(S1)+(.1,0.25,0)$) {\scalebox{1.5}{$\omega_{\scaleto{1}{3pt}}$}};
\node at ($(S2)+(-0.2,0.2,0)$) {\scalebox{1.5}{$\omega_{\scaleto{4}{3pt}}$}};
\node at ($(S3)+(0,-0.25,0)$) {\scalebox{1.5}{$\omega_{\scaleto{3}{3pt}}$}};
\node at ($(S4)+(.1,-0.25,0)$) {\scalebox{1.5}{$\omega_{\scaleto{2}{3pt}}$}};
\node at (.22,0.15,0) {\scalebox{1.5}{$\omega_{\scaleto{m}{2pt}}$}};
\node at (0,2,0) {\scalebox{1.5}{$e^{\scaleto{0}{3pt}}_{\scaleto{0}{3pt}}$}};
\node at (0,-2,0) {\scalebox{1.5}{$e^{\scaleto{1}{3pt}}_{\scaleto{0}{3pt}}$}};
\node at (2,0,0) {\scalebox{1.5}{$e^{\scaleto{0}{3pt}}_{\scaleto{1}{3pt}}$}};
\node at (-2,0,0) {\scalebox{1.5}{$e^{\scaleto{1}{3pt}}_{\scaleto{1}{3pt}}$}};
\node at (2.1,-1.6,0) {\scalebox{1.5}{$p$}};
\node at (-1.65,2.1,0) {\scalebox{1.5}{$q$}};
\end{tikzpicture} 
\caption{(Color online) Given the state  $\omega\in\mathbf{S}_{\bm\square}$, the sets $\{\omega_{x_0x_1}\}$ and $\{\omega^\prime_{x_0x_1}\}$ denote the encoding states when Alice is restricted to apply $\mathrm{RaRe}_{\bm\square}$ and $\mathrm{BiSto}_{\bm\square}$, respectively. For decoding $x_0$, Bob performs the measurement  $M_{13}\equiv\left\{e_0^0:=e_1,~e_0^1:=e_3\right\}$ and and guess the bit value as $i\in\{0,1\}$ when the effect $e_0^i$ clicks. For $x_1$ similar strategy is applied with the measurement $M_{24}\equiv\left\{e_1^0:=e_2,~e_1^1:=e_4\right\}$.}
\vspace{-.5cm}
\label{fig5}
\end{figure}
\begin{align}\label{Rare_encoding}
\left\{\!\begin{aligned}
\omega_{00}&:=0.5\left(\leftindex_3^4{\square}_{2}^{1}+\leftindex_3^2{\square}_{4}^{1}\right)(\omega)=(+\alpha,+\alpha,1)^\intercal,\\
\omega_{01}&:=0.5\left(\leftindex_4^1{\square}_{3}^{2}+\leftindex_2^1{\square}_{3}^{4}\right)(\omega)=(-\alpha,+\alpha,1)^\intercal,\\
\omega_{10}&:=0.5\left(\leftindex_2^3{\square}_{1}^{4}+\leftindex_4^3{\square}_{1}^{2}\right)(\omega)=(+\alpha,-\alpha,1)^\intercal,\\
\omega_{11}&:=0.5\left(\leftindex_1^2{\square}_{4}^{3}+\leftindex_1^4{\square}_{2}^{3}\right)(\omega)=(-\alpha,-\alpha,1)^\intercal
\end{aligned}\right\};
\end{align}
where, $\alpha:=1/2[(2p-1)+(2q-1)]$, and Bob's decoding is given by
\begin{align}
\left\{\!\begin{aligned}
x_0:~M_0\equiv M_{13}\equiv\left\{e_0^0:=e_1,~e_0^1:=e_3\right\}\\
x_1:~M_1\equiv M_{24}\equiv\left\{e_1^0:=e_2,~e_1^1:=e_4\right\}
\end{aligned}\right\}~.\label{boxd}    
\end{align}
The average success probability thus becomes 
\begin{align}
P_\omega[\mathrm{RaRe}_{\bm\square}]=&\frac{1}{8}(\sum_{x_0,x_1,y}T_{x_0x_1}(\omega).e^y_{x_y})\nonumber\\
&=\frac{1}{2}(1+\alpha)\\&=\frac{1}{4}(2+|2p-1|+|2q-1|)
\end{align}
For an arbitrary state $\omega=(2p-1,2q-1,1)^\intercal\in\mathbf{S}_{\bm\square}$, with $p,q\in[0,1]$, the success probability thus reads as $P_\omega[\mathrm{RaRe}_{\bm\square}]=\frac{1}{4}(2+|2p-1|+|2q-1|$.
\\
The optimality of the protocol is established by exploiting the geometry of the RaRe-accessible region and optimizing over all admissible encoding transformations and decoding measurements. Alice encodes her input string $x_0x_1$ by applying an appropriate reversible transformation $T_{x_0x_1}$ to the initial state $\omega$, and Bob performs a two outcome measurements $M_y=\{e_y^0,e_y^1\}$ according to his input $y$. Consequently, determining the optimal average success probability amounts to solving the following optimization problem
\begin{equation}
\begin{split}
\max_{\substack{T_{x_0x_1},\,e^y_i}}\quad&
\frac{1}{8}\sum_{x_0,x_1,y}T_{x_0x_1}(\omega)\cdot e^y_{x_y}\\
\text{s.t.}\quad&
T_{x_0x_1}\in\mathrm{RaRe}_{\square},~
e^y_i\in\mathbf{E},~ \sum_i e^y_i=u.
\end{split}
\end{equation}
Every admissible measurement admits a convex decomposition into extremal measurements. Since the objective function is linear in the measurement effects, the optimization over the convex measurement set achieves its optimum at an extremal point. Hence, without loss of generality, the optimal decoding strategy can be restricted to extremal measurements.  Accordingly, we restrict our attention to decoding strategies in which Bob employs two distinct extremal measurements, and optimize over all such choices. In the square-bit GPT, there exist only two extremal binary measurements, $M_0\equiv M_{13}\equiv\left\{e_0^0:=e_1,~e_0^1:=e_3\right\}$ and $M_1\equiv M_{24}\equiv\left\{e_1^0:=e_2,~e_1^1:=e_4\right\}$.
We therefore consider decoding strategies in which Bob employs the measurement $M_y$ to decode the bit indexed by $y\in\{0,1\}$. Once the decoding measurements are fixed, the optimization over the encoding channels can be carried out independently for each input string. In particular, for a given initial state $\omega=(2p-1,2q-1,1)^\intercal\in\mathbf{S}_{\square},$ and the input string $x_0x_1=00$, the optimal reversible transformations are those that maximize the linear functional $\frac{1}{2}\bigl(e_0^0+e_0^1\bigr).\!\left(T_{00}(\omega)\right)$
over the set of reversible transformations.
Therefore, the optimum is attained by the two extremal reversible transformations
$\leftindex_3^4{\square}_{2}^{1}$ and
$\leftindex_3^2{\square}_{4}^{1}$, each of which yields the maximum value \(p+q\). The linearity of the optimization objective implies that any convex combination of these two extremal transformations remains optimal and achieves the optimal average success probability. For the sake of symmetric encoding strategy we have considered the equal convex mixture of these two extremal transformation  {\it i.e.} $T_{00}=\frac{1}{2}\left(\leftindex_3^4{\square}_{2}^{1}+\leftindex_3^2{\square}_{4}^{1}\right)$. The optimality of the other reversible transformation follows from the same argument. Hence the encoding reversible transformation in Eq. (\ref{Rare_encoding}) and decoding measurements in Eq. (\ref{boxd}) are proven to be optimal under RaRe transformation.\\
{\it Optimal success under $\mathrm{BiSto}_{\bm\square}$.--} In this case, Alice's encoding ($x_0x_1\mapsto\omega^\prime_{x_ox_1}$) is given by: 
\begin{align}
\left\{\!\begin{aligned}
\omega^\prime_{00}&:=\leftindex_{3}^{4}{\square}_{2}^{1}\circ\leftindex_{23}^{}{\square}_{}^{14}(\omega)=(+\beta,+\beta,1)^\intercal,\\
\omega^\prime_{01}&:=\leftindex_{4}^{1}{\square}_{3}^{2}\circ\leftindex_{23}^{}{\square}_{}^{14}(\omega)=(-\beta,+\beta,1)^\intercal,\\
\omega^\prime_{10}&:=\leftindex_{2}^{3}{\square}_{1}^{4}\circ\leftindex_{23}^{}{\square}_{}^{14}(\omega)=(+\beta,-\beta,1)^\intercal,\\
\omega^\prime_{11}&:=\leftindex_{1}^{2}{\square}_{4}^{3}\circ\leftindex_{23}^{}{\square}_{}^{14}(\omega)=(-\beta,-\beta,1)^\intercal\end{aligned}\right\};
\end{align}
where, $\beta=\max\{|2p-1|,|2q-1|\}$, and Bob follows the decoding (\ref{boxd}). Accordingly, the success probability becomes
\begin{align}
P_\omega[\mathrm{BiSto}_{\bm\square}]=&\frac{1}{8}(\sum_{x_0,x_1,y}T_{x_0x_1}(\omega).e^y_{x_y})\nonumber\\
&=\frac{1}{2}(1+\beta)=\frac{1}{2}(1+\max\{|2p-1|,|2q-1|\}).
\end{align}
Manifestly, $P_\omega[\mathrm{RaRe}_{\bm\square}]\le P_\omega[\mathrm{BiSto}_{\bm\square}]$, and the equality holds only for the states with $|2p-1|=|2q-1|$. 
\end{proof}
One might expect a similar gap as of Theorem \ref{theo4} with higher dimensional quantum systems by considering the corresponding higher-level RAC tasks \cite{Tavakoli2015}. However, such an intuition is not true as the state transformation criteria under quantum RaRe channels and under the set of unital channels boil down to the same majorization conditions \cite{Uhlmann1970,Gour2015}.\\\\
{\it Discussion on Asymptotic Birkhoff-violation in non quantum GPT.--}\label{Assymptotic}
Smolin {\it et al.} in Ref.\cite{Smolin2005} raised an interesting question whether Birkhoff theorem holds in asymptotic setup. The distance between a channel $\Phi$ from the set of random reversible channel RaRe($\mathcal{H}$) can be measured by $\mathcal{D}(\Phi,\mbox{RaRe}(\mathcal{H})):=\inf \{\mathcal{D}(\Phi,\Psi):\Psi \in\mbox{RaRe}(\mathcal{H})\}$, where metric $\mathcal{D}(\Phi,\Psi)$ is given  by the `diamond norm' of ($\Phi-\Psi$). The question of asymptotic Birkhoff theorem thus boils down to whether $\lim_{n \to \infty} \mathcal{D}\left(\Phi^{\otimes n},\mbox{RaRe}(\mathcal{H}^{\otimes n})\right)\to0$, for all unital channels $\Phi$. It has been shown that the theorem does not hold in this asymptotic set up as well \cite{Haagerup2011,Haagerup2015}. In the following we show a weaker version of the asymptotic violation for the square  bit theory.
\begin{theorem}\label{theo5}
The square-bit theory violates the Birkhoff theorem for every finite but arbitrarily large tensor power. In particular, there exists a bistochastic channel
$T\in\mathrm{BiSto}_{\square}$ such that
$T^{\otimes n}\notin\mathrm{RaRe}_{\square^{\otimes n}}$,$~\forall\, n\in\mathbb{N}_{+},$
where $\mathrm{RaRe}_{\square^{\otimes n}}$ denotes the set of all random reversible channels acting on the composite $n$-square-bit system.
\end{theorem}
\begin{proof}
Consider the extreme bistochastic map $T = [\gamma_1^{+}, \gamma_1^{+}] \in \mathrm{BiSto}_{\square},$ corresponding to the first element in Table~I. The action of its $n$-fold tensor product on the extremal states $\omega_i^{\otimes n}$ yields $T^{\otimes n}(\omega_i^{\otimes n}) = \omega_1^{\otimes n},$ where $i \in \{1,2\}$\footnote{Here, we use the fact that $\omega_{i^{\otimes n}}$ is an extremal state of both the minimal and maximal tensor-product state spaces}. Assume, by contradiction, that $T^{\otimes n} = \sum_j p_j R_j \in \mathrm{RaRe}_{\square^{\otimes n}},$ where each $R_j$ is an extremal reversible channel on the composite state space $S_{\square^{\otimes n}}$ of $n$ square-bit systems $(
R: S_{\square^{\otimes n}}\rightarrow S_{\square^{\otimes n}})$.
The condition $T^{\otimes n}(\omega_i^{\otimes n}) = \omega_1^{\otimes n}$ implies that $R_j(\omega_i^{\otimes n}) = \omega_1^{\otimes n} ~ \text{for all } j \text{ such that } p_j>0,   $ since the states $\omega_i^{\otimes n}$ are extreme points of $S_{\square^{\otimes n}}$. 
However, this leads to a contradiction: two distinct extreme states $\omega_1^{\otimes n}$ and $\omega_2^{\otimes n}$ are mapped to the same state by reversible transformations $R_j$. Hence, the assumption $ T^{\otimes n} \in \mathrm{RaRe}_{\square^{\otimes n}}$ is false.
\end{proof}
At this point, one needs to clarify the kind of composition ${\bf S}_{\bm\square^{\otimes n}}$ considered in Theorem \ref{theo5}. Different compositions are possible for two square-bits, which range in between two extremes:
\begin{subequations}
\begin{align}
\bigotimes_{\min}:\left\{\!\begin{aligned}
{\bf S}^{\min}_{[2]}&:=\mbox{ConvHul}\{\omega_i\otimes\omega_j\}\\
{\bf E}^{\min}_{[2]}&:=\mbox{ConvHul}\{e_i\otimes e_j,~ E_k\}
\end{aligned}\right\},~\\
\bigotimes_{\max}:\left\{\!\begin{aligned}
{\bf S}^{\max}_{[2]}&:=\mbox{ConvHul}\{\omega_i\otimes\omega_j,~\Omega_k\}\\
{\bf E}^{\max}_{[2]}&:=\mbox{ConvHul}\{e_i\otimes e_j\}
\end{aligned}\right\},
\end{align}
\end{subequations}
where $i,j\in\{1,\cdots,4\}$ and $k\in\{1,\cdots,8\}$ with $E_K$ and $\Omega_k$ respectively denoting entangled effects and entangled states (see \cite{DallArno2017} for further details). A similar concept extends for multiple square-bits. Our Theorem \ref{theo5} holds for both the minimal and maximal tensor product compositions, denoted by $\otimes_{\min}$ and $\otimes_{\max}$, respectively. A generalization of Theorem \ref{theo5} is proven below for the evengon models where normalized state-space ${\bf S}_{2k}$ is specified as the convex hull of its $2k$ extreme points \cite{Janotta2011}:
\begin{align*}
{\bf S}_{2k}\equiv \mbox{ConvHull}\left\{\omega_i:=\left(\cos{\frac{\pi i}{k}},\sin{\frac{\pi i}{k}},1\right)^\intercal\right\}_{i=1}^{2k};
\end{align*}
where $k\ge2$ (see Fig. \ref{fig6}).
\begin{figure}[h!]
\centering
\begin{tikzpicture}[scale=2.4]
\coordinate (A1) at (0.5,0.866025,0.);
\coordinate (A2) at (0.,1.,0.);
\coordinate (A3) at (-0.5,0.866025,0.);
\coordinate (A4) at (-0.866025,0.5,0);
\coordinate (A5) at (-1.,0.,0.);
\coordinate (A6) at (-0.866025,-0.5,0.);
\coordinate (A7) at (-0.5,-0.866025,0.);
\coordinate (A8) at (0.,-1.,0.);
\coordinate (A9) at (0.5,-0.866025,0.);
\coordinate (A10) at (0.866025,-0.5,0);
\coordinate (A11) at (1.,0.,0.);
\coordinate (A12) at (0.866025,0.5,0.);
\coordinate (B1) at (-0.333333, -0.57735, 0);
\coordinate (B2) at (-0.4, -0.69282, 0);
\coordinate (B3) at (0.333333, 0.57735, 0);
\coordinate (B4) at (0.4, 0.69282, 0);
\draw [thick,fill=blue!0, fill opacity=0.3] (A1)--(A2);
\draw [thick,fill=blue!0, fill opacity=0.3] (A7)--(A8);
\draw [thick,fill=blue!0, fill opacity=0.3] (A2)--(A3);
\draw [thick,fill=blue!0, fill opacity=0.3] (A1)--(A12);
\draw [thick,fill=blue!0, fill opacity=0.3] (A7)--(A6);
\draw [thick,fill=blue!0, fill opacity=0.3] (A8)--(A9);
\draw [thick,fill=blue!0, fill opacity=0.3] (A1)--(A7);
\draw[arrows = {-Latex[width=5pt, length=10pt]},thin] (A2)--(A1);
\draw[arrows = {-Latex[width=5pt, length=10pt]},thin] (A8)--(A7);
\draw[arrows = {-Latex[width=5pt, length=10pt]},thin] (A9)--(B1);
\draw[arrows = {-Latex[width=5pt, length=10pt]},thin] (A6)--(B2);
\draw[arrows = {-Latex[width=5pt, length=10pt]},thin] (A12)--(B4);
\draw[arrows = {-Latex[width=5pt, length=10pt]},thin] (A3)--(B3);
\draw [thick,dotted,fill=blue!0, fill opacity=0.3] (A12)--(A11);
\draw [thick,dotted,fill=blue!0, fill opacity=0.3] (A3)--(A4);
\draw [thick,dotted,fill=blue!0, fill opacity=0.3] (A5)--(A6);
\draw [thick,dotted,fill=blue!0, fill opacity=0.3] (A9)--(A10);
\shade[ball color=green1] (0,0,0) circle (.04cm);
\draw[thick] (A1) edge [loop above] (A1);
\draw[thick] (A7) edge [loop below] (A7);
\shade[ball color=purple] (A1) circle (.04cm);
\shade[ball color=blue] (A2) circle (.04cm);
\shade[ball color=blue] (A3) circle (.04cm);
\shade[ball color=blue] (A6) circle (.04cm);
\shade[ball color=purple] (A7) circle (.04cm);
\shade[ball color=blue] (A8) circle (.04cm);
\shade[ball color=blue] (A9) circle (.04cm);
\shade[ball color=blue] (A12) circle (.04cm);
\shade[ball color=abc] (B1) circle (.04cm);
\shade[ball color=abc] (B2) circle (.04cm);
\shade[ball color=abc] (B3) circle (.04cm);
\shade[ball color=abc] (B4) circle (.04cm);
\node at ((.2,0,0) {\scalebox{1.5}{$\omega_m$}};
\node at ($(A1)+(.16,0.14,0)$) {\scalebox{1.5}{$\omega_1$}};
\node at ($(A2)+(-.14,0.10,0)$) {\scalebox{1.5}{$\omega_2$}};
\node at ($(A3)+(-.14,0.14,0)$) {\scalebox{1.5}{$\omega_3$}};
\node at ($(A6)+(-.16,-0.08,0)$) {\scalebox{1.5}{$\omega_k$}};
\node at ($(A7)+(-.24,-0.14,0)$) {\scalebox{1.5}{$\omega_{k+1}$}};
\node at ($(A8)+(.16,-0.12,0)$) {\scalebox{1.5}{$\omega_{k+2}$}};
\node at ($(A9)+(.19,-0.14,0)$) {\scalebox{1.5}{$\omega_{k+3}$}};
\node at ($(A12)+(.24,-0.10,0)$) {\scalebox{1.5}{$\omega_{2k}$}};
\end{tikzpicture} 
\caption{(Color online) state-space ${\bf S}_{2k}$ of evengon-bit. The action of the bistochastic transformation $T$ on extreme states are depicted.}
\vspace{-.5cm}
\label{fig6}
\end{figure}
\begin{theorem}
For all $k \ge 2$, the evengon-bit theory given by the state space ${\bf S}_{2k}$ violates Birkhoff's theorem for arbitrary large but finite copies.
\end{theorem}
\begin{proof}
The completely mixed state $\omega_m\in{\bf S}_{2k}$ is given by $\omega_m:=(0,0,1)^\intercal$, which is equal mixture of all the extreme points. This can also be expressed as $\omega_m=1/2(\omega_i+\omega_{i+k})$, $\forall ~i\in\{1,2,\cdots,k\}$. Consider the linear transformation $T$ (see Fig.\ref{fig6}) given by
\begin{align}
\left\{\!\begin{aligned}
T(\omega_1)=\omega_1,&~T(\omega_2)=\omega_1,\\T(\omega_{k+1})=\omega_{k+1},&~T(\omega_{k+2})=\omega_{k+1}\end{aligned}\right\}.
\end{align}  
It suffices to define the action on the above $4$ states to define this map as these $4$ states for a spanning set for every evengon-bit. The action of $T$ on completely mixed state yields
\begin{align}
T(\omega_m)=T\left(\frac{1}{2}(\omega_1+\omega_{k+1})\right)\nonumber\\
=\frac{1}{2}(\omega_1+\omega_{k+1})=\omega_m,
\end{align}
ensuring $T$ is indeed a bistochastic map. The asymptotic violation of Birkhoff's theorem is analogous to the proof of Theorem \ref{theo5} and hinges on the crucial fact that $T(\omega_1)=T(\omega_2)=\omega_1$.
\end{proof}

\color{black}
{\it Entropy of a GPT State.--}
The quantity $S_{vN}(\omega)$ can be defined as the entropy of the GPT state $\omega\in{\bf S}_{\bm\square}$. Notably, it carries the characteristics of thermodynamic entropy as desired in von Neumann’s thought experiment \cite{vonNeumann1955} (see \cite{Petz2001} for a modern description of the same experiment). According to this definition, the boundary states of ${\bf S}_{\bm\square}$ have zero entropy. At this point, it is important to note that the notion of pure states in classical and quantum theory perfectly matches with the notion of extremality of state-space. However, this needs not to be the case for an arbitrary GPT. While extremality is a geometric concept, the notion of purity should be defined from an operational perspective. Recall that measurements in GPTs are the probing mechanism for accessing information about the system. Given a measurement $M \equiv \{e_i\}_{i=1}^k$, it might allow a finer refinement $\tilde{M} \equiv \{\tilde{e}_i\}_{i=1}^{\tilde{k}}$, where $\tilde{k}>k$ and suitable grouping of $\tilde{e}_i$'s reproduces the outcome statistics of $e_i$'s. A measurement is said to be fine-grained if it does not allow any nontrivial refinement. Such a measurement extracts information from the states in optimal possible way. One can now come up with an operationally motivated definition for the pure states (i.e., the states of maximal knowledge).
\begin{definition}
A state $\omega\in\mathbf{S}$ is said to be pure, if there exists at least one fine-grained measurement $M = \{e_i\}_{i=1}^k$ such that one of its effects $e_{i^\star}$ filters the state deterministically, i.e., $e_{i^\star}(\omega) = 1$. 
\end{definition}
This definition captures all the extremal classical and quantum states as pure. The boundary states of ${\bf S}_{\bm\square}$ also become pure according to this definition, and hence they should be assigned zero entropy. Importantly, this definition possess the crucial aspect of the device called `semipermeable membranes', introduced by von Neumann to capture the notion of `thermodynamic differentness' \cite{Szilard1925,McKeown1927,Klein1967}. While most of the literature identify purity as extremality \cite{Barnum2010,Kimura2010,Brunner2014,Krumm2017}, we find that our notion of entropy $S_{vN}(\star)$ perfectly match with the one proposed in \cite{Short2010}. However, a complete investigation of entropy also demands a careful analysis of measurement update in GPT, which we leave for future study.
{\it Discussion}
Our analysis of the Birkhoff-von Neumann theorem in the framework of GPTs shows that its violation is not an exclusive feature of quantum theory, but rather a generic structural property of a broad class of operational theories, including the square-bit model and all even-polygon GPTs. Such violations give rise to consequences that have no direct analogue in quantum theory. In particular, whereas in quantum mechanics bistochastic maps and convex mixtures of reversible transformations induce identical state transformation criteria, this equivalence fails in more general GPTs. Consequently, these two classes of dynamics can lead to inequivalent notions of state convertibility, signaling a fundamental departure from the quantum case.
\color{black}
Although microcanonical theories satisfying the sharp theory of purity \cite{Chiribella2017} are known to exhibit Birkhoff violation, prior work has largely focused on its dynamical aspects. In contrast, our results establish its operational significance in information-processing tasks. Specifically, we show that the inequivalence between unital and random reversible (RaRe) transformations—arising directly from Birkhoff violation—translates into a performance advantage in communication tasks such as Random Access Codes. This identifies Birkhoff violation as an operational resource, which leads to enhanced information-processing capabilities.
In this direction, we identify distinct entropic quantities that are monotonic under bistochastic maps and RaRe dynamics, respectively. We also define a consistent notion of entropy that retains the essential features of thermodynamic entropy in the sense of von Neumann's thought experiment.
At a structural level, we establish asymptotic Birkhoff violation for all even polygon theories. Establishing a general analytical understanding of this phenomenon remains an open problem. More broadly, it would be of interest to characterize those GPTs that, like quantum theory, retain an equivalence between state transformation criteria under random reversible and bistochastic operations despite violating the Birkhoff theorem. Such theories may occupy a distinguished position within the GPT landscape. Looking ahead, these results suggest several directions for future investigation. A systematic classification of operational theories based on the relationship between reversible controllability, stochastic mixing, and state convertibility may provide new insight into the structure of physical theories. In particular, understanding the precise conditions under which Birkhoff-type equivalences are recovered could help isolate principles that single out quantum theory from more general probabilistic models. Furthermore, the operational advantages identified here raise the possibility that Birkhoff violation may serve as a resource in broader information-processing tasks beyond Random Access Codes. Finally, the interplay between entropy, irreversibility, and operational constraints in GPTs points toward a more general formulation of thermodynamics, potentially extending beyond the standard quantum framework.

{\bf Acknowledgment:} SGN acknowledges support from the CSIR project 09/0575(15951)/2022-EMR-I.  RKP acknowledges support from the EU (CHISTERA MoDIC) and from the National Research, Development and Innovation Office NKFIH ($2023-1.2.1-$ERA\_NET$-2023-00009$). PG acknowledges financial support from the project entitled “Technology Vertical - Quantum Communication” under the National Quantum Mission (NQM) of the Department of Science and Technology (DST) (Sanction Order No. DST/QTC/NQM/QComm/2024/2 (G)). SS acknowledges financial support from the European Union (ERC StG ETQO, Grant Agreement no. 101165230). Views and opinions ex-pressed are however those of the author(s) only and do not necessarily reflect those of the European Union or the European Research Council. Neither the European Union nor the granting authority can be held responsible for them. MA acknowledges the funding supported by the European Union  (Project QURES- GA No. 101153001). MB acknowledges the financial support through the National Quantum Mission (NQM) of the Department of Science and Technology, Government of India.

\providecommand{\newblock}{}

\end{document}